\documentclass[12pt]{article}

\usepackage{lipsum, layouts, kantlipsum}
\usepackage{amsfonts}
\usepackage{amsmath}
\usepackage{amssymb}
\usepackage{amsthm}
\usepackage{bm}
\usepackage{bbm}
\usepackage{booktabs}
\usepackage[font=bf,justification=centering]{caption}
\usepackage{enumerate}
\usepackage{float}
\usepackage[top=1in, bottom=1in, left=1in, right=1in]{geometry}
\usepackage{graphicx}
\usepackage[hidelinks,bookmarks]{hyperref}
\usepackage{natbib}
\usepackage{setspace}
\usepackage{tabularx}
\usepackage{titling}

\usepackage{bibunits}

\hypersetup{pdfstartview={XYZ null null 1.00}} 
\allowdisplaybreaks
\theoremstyle{definition}

\theoremstyle{definition}

\theoremstyle{plain}

\theoremstyle{plain}
\newtheorem{proposition}{Proposition}[section]
\newtheorem{corollary}{Corollary}[section]

\title{\vspace{-2.0cm}Sign Restrictions and Supply-demand Decompositions of Inflation}
\author{Matthew Read\thanks{I thank Benjamin Beckers, Robin Braun, Jarkko~J\"{a}\"{a}skel\"{a}, Benjamin Wong and seminar participants at the Reserve Bank of Australia and University of Adelaide, three anonymous referees and Marco Del Negro for helpful comments. The views expressed in this paper are my own and do not necessarily reflect the views of the Reserve Bank of Australia. The author has no conflicts of interest to disclose. Correspondence: Matthew Read, Economic Research Department, Reserve Bank of Australia, GPO Box 3947, Sydney, NSW, 2001, Australia. Email address: readm@rba.gov.au.}}

\defaultbibliography{SupplyDemand.bib}        
\defaultbibliographystyle{ecta}

\begin{document}

\begin{bibunit}

\maketitle

\begin{abstract}
    Sign restrictions on the slopes of supply and demand curves are often used to identify historical decompositions in structural vector autoregressions. I show that the identifying power of these restrictions depends on both reduced-form parameters and realised forecast errors. Consequently, unlike many other structural objects, the strength of identification cannot be assessed from reduced-form parameters alone. Empirically, identified sets for historical decompositions of US inflation are typically largely uninformative, both in aggregate and in most expenditure categories. Existing inflation decompositions are therefore sensitive to auxiliary assumptions used to select among observationally equivalent models. 
\end{abstract}



\section{Introduction}
\label{sec:intro}

Distinguishing supply- from demand-driven changes in prices and quantities is important because the shocks may have different policy implications. A prominent recent example is the debate over the sources of the pandemic-era surge and subsequent decline in inflation.\footnote{US examples include \cite{Ball_Leigh_Mishra_2022}, \cite{Banbura_Bobeica_MartinezHernandez_2023}, \cite{Beaudry_Hou_Portier_2025}, \cite{Bernanke_Blanchard_2025} and \cite{Eickmeier_Hofmann_2022}.}

A common strategy for decomposing changes in variables into shock contributions is to estimate a structural vector autoregression (SVAR). When decomposing changes in prices and quantities into contributions from supply and demand shocks, one approach to identification is to impose sign restrictions on the slopes of supply and demand curves.\footnote{Following \cite{Faust_1998}, \cite{Canova_DeNicolo_2002} and \cite{Uhlig_2005}, sign restrictions are widely used in the broader SVAR literature.} Economic theory typically implies that supply curves are upward sloping and demand curves are downward sloping. Hence, the appeal of these restrictions is that they are uncontroversial. There is, however, a cost: sign restrictions set-identify shock decompositions, so a range of decompositions are observationally equivalent.

This paper examines the conditions under which these sign restrictions are informative about historical decompositions of observed price changes -- an important question given their use empirically in estimating the drivers of inflation. The paper makes two main contributions. First, I analytically characterise identified sets for historical decompositions and show how their informativeness depends on features of the data. Second, I estimate identified sets for supply-demand decompositions of US inflation, transparently showing what the sign restrictions alone reveal about the contributions of shocks.

The theoretical discussion builds on work examining the identifying power of sign restrictions in supply-demand systems. \cite{Baumeister_Hamilton_2015} use a supply-demand model to illustrate the influence of the `uniform' prior on Bayesian inference in sign-restricted SVARs, and derive identified sets for the price elasticities of supply and demand (p. 1975--1976). \cite{Uhlig_2017} illustrates the use of sign restrictions when identifying the slopes of supply and demand curves.\footnote{\cite{Leamer_1981} contains a similar discussion in the context of maximum likelihood estimation of simultaneous-equation systems subject to inequality constraints.} \cite{Inoue_Kilian_2026} illustrate how identification strength depends on reduced-form parameters, focusing on impulse responses. Like these papers, I consider a bivariate VAR in prices and quantities identified by supply-demand sign restrictions, but I focus on historical decompositions, motivated by recent empirical applications.

In the analytically tractable case where interest is in decomposing one-step-ahead forecast errors, I show that the identifying power of the sign restrictions depends on both the reduced-form correlation between price and quantity forecast errors, $\rho$, and the realisations of those forecast errors. When $\rho$ is close to zero, decompositions are necessarily uninformative: observed price changes can be explained almost entirely by either shock. When $|\rho|$ is large, a sharp decomposition is possible, but only if the realised forecast errors align with their historical relationship. Consequently, $\rho$ is not a sufficient statistic for assessing the identifying power of the restrictions with respect to historical decompositions. This distinguishes historical decompositions from objects such as impulse responses, elasticities and forecast error variance decompositions, for which identification is governed solely by reduced-form parameters (e.g. \citealt{Baumeister_Hamilton_2015,Uhlig_2017,Inoue_Kilian_2026}).

Empirically, recent papers have used bivariate models to estimate historical decompositions of inflation, exploiting sign restrictions on the slopes of supply and demand curves. A prominent example is \cite{Giannone_Primiceri_2024_nber}, who estimate the contributions of supply and demand shocks to inflation in the United States and euro area, finding that post-pandemic inflation was predominantly driven by demand shocks.\footnote{Other examples include \cite{DellaChang_Jansen_Pagliacci_2023}, \cite{Liepnieks_Staehr_Tkacevs_2026} and \cite{Bergholt_Furlanetto_Vaccaro-Grange_2024}.} A feature of these papers is that they work with a single decomposition chosen from a set of observationally equivalent decompositions. It is therefore unclear to what extent results are driven by the auxiliary assumptions implicit in selecting a single decomposition. I therefore revisit the model of \cite{Giannone_Primiceri_2024_nber}, directly estimating identified sets for the contributions of aggregate supply and demand shocks to post-pandemic inflation forecast errors. The resulting sets are wide: for example, at the 2022:Q2 inflation peak, the estimated supply contribution ranges from near zero to almost the entire forecast error.

Motivated by the inflation decomposition in \cite{Shapiro_2026}, which exploits disaggregated data on prices and quantities, I also estimate identified sets for historical decompositions of US inflation in different expenditure categories.\footnote{The inflation decomposition in \cite{Shapiro_2026} has informed policymakers' assessments of the economic outlook (e.g. \citealt{Lane_2023,Kugler_2024}). It has also been applied in various other settings (e.g. \citealt{Goncalves_Koester_2022,Chen_Tombe_2023,Firat_Hao_2023}).} The restrictions are largely uninformative in most expenditure categories and periods, although they sharply identify the drivers of inflation in some categories. For example, for `food produced and consumed on farms', $\rho \approx -1$, so price and quantity changes approximately trace out a short-run demand curve and inflation is predominantly attributed to supply shocks.

Sign restrictions on the slopes of supply and demand curves, on their own, may not deliver unambiguous conclusions about the contributions of supply and demand shocks to changes in prices (or quantities). This is the case when decomposing US inflation, both in aggregate and in most expenditure categories. In these cases, any additional assumptions inherent in selecting a single model or decomposition -- such as the specification of a Bayesian prior distribution -- have a strong influence on conclusions about the contributions of shocks.

\medskip

\noindent\textbf{Outline.} Section~\ref{sec:SVAR} introduces the model, Section~\ref{sec:informativeness} develops the identification results, and Section~\ref{sec:empirical} presents the empirical applications. The Online Appendix contains derivations and implementation details.

\medskip

\noindent\textbf{Notation.} Vectors and matrices are bold. For a matrix $\mathbf{X}$, $\mathrm{vec}(\mathbf{X})$ stacks its columns and $\mathrm{vech}(\mathbf{X})$ stacks its lower-triangular elements. $\mathbf{e}_{i}$ is column $i$ of the $2\times 2$ identity matrix, $\mathbf{I}_2$.

\section{SVAR and Decompositions}
\label{sec:SVAR}

This section introduces the bivariate SVAR and the historical decomposition.

\subsection{SVAR and orthogonal reduced form}

Assume that $\mathbf{y}_{t} = (p_t,q_t)'$ contains data on prices $p_t$ and quantities $q_t$, and is generated by the SVAR($p$) process:
\begin{equation}\label{eq:svar}
    \mathbf{A}_{0}\mathbf{y}_{t} = \mathbf{A}_{+}\mathbf{x}_{t} + \bm{\varepsilon}_t,
\end{equation}
where: $\mathbf{A}_{0}$ is an invertible matrix with positive diagonal elements (a sign normalisation); $\mathbf{x}_{t} = (\mathbf{y}_{t-1}',\ldots,\mathbf{y}_{t-p}',1)'$; $\mathbf{A}_{+} = (\mathbf{A}_{1},\ldots,\mathbf{A}_{p},\mathbf{a})$; and $\bm{\varepsilon}_t = (\varepsilon_{1t},\varepsilon_{2t})'$ are the structural shocks, which have zero mean and identity variance-covariance matrix.

The model's orthogonal reduced-form parameterisation is:
\begin{equation}
    \mathbf{y}_t = \mathbf{B}\mathbf{x}_{t} + \bm{\Sigma}_{tr}\mathbf{Q}\bm{\varepsilon}_t,
\end{equation}
where: $\mathbf{B} = (\mathbf{B}_{1},\ldots,\mathbf{B}_{p},\mathbf{b}) $ is a matrix of reduced-form coefficients; $\bm{\Sigma}_{tr}$ is the lower-triangular Cholesky factor of the reduced-form innovation variance-covariance matrix $\bm{\Sigma}=\mathbb{E}(\mathbf{u}_{t}\mathbf{u}_{t}')$ with $\mathbf{u}_{t} = (u_{pt},u_{qt})'= \mathbf{y}_{t} - \mathbf{B}\mathbf{x}_{t}$; and $\mathbf{Q}$ is an orthonormal matrix, with $\mathcal{O}(2)$ the set of all such matrices. Let~$\bm{\phi} = (\mathrm{vec}(\mathbf{B})',\mathrm{vech}(\bm{\Sigma}_{tr})')'$ collect the reduced-form parameters. The two parameterisations are related by $\mathbf{B} = \mathbf{A}_{0}^{-1}\mathbf{A}_{+}$, $\bm{\Sigma} = \mathbf{A}_{0}^{-1}(\mathbf{A}_{0}^{-1})'$ and $\mathbf{Q} = \bm{\Sigma}_{tr}^{-1}\mathbf{A}_{0}^{-1}$.

Absent identifying restrictions, every $\mathbf{Q} \in \mathcal{O}(2)$ is consistent with $\bm{\phi}$, so $\mathbf{Q}$ and the structural parameters are set identified (e.g. \citealt{Uhlig_2005}). Sign restrictions restrict $\mathbf{Q}$ to an identified set, inducing identified sets for structural objects like historical decompositions. 

\subsection{Historical decomposition}
\label{subsec:historicaldecomp}

The historical decomposition is the contribution of a shock to the observed unexpected change (forecast error) in a variable over some horizon (e.g. \citealt{Antolin-Diaz_Rubio-Ramirez_2018,Baumeister_Hamilton_2018}). Define the reduced-form impulse response $\mathbf{C}_h$ via the recursion $\mathbf{C}_h = \sum_{l=1}^{\mathrm{min}\{h,p\}}\mathbf{B}_{l}\mathbf{C}_{h-l}$ for $h\geq 1$ with $\mathbf{C}_0=\mathbf{I}_2$. The contribution of shock $j$ to the forecast error in variable $i$ between periods $t$ and $t+h$ is:
\begin{align}
    H_{i,j,t,t+h}  &= \mathbb{E}\left(y_{i,t+h} | \left\{\varepsilon_{j,\tau}\right\}_{t \leq \tau \leq t+h }, \left\{\mathbf{y}_{\tau}\right\}_{-\infty < \tau \leq t-1}\right) - \mathbb{E}\left(y_{i,t+h} | \left\{\mathbf{y}_{\tau}\right\}_{-\infty < \tau \leq t-1}\right) \label{eq:histdecomp} \\
    &= \sum_{l=0}^{h}\mathbf{c}_{il}'(\bm{\phi})\mathbf{q}_{j} \mathbf{q}_{j}'\bm{\Sigma}_{tr}^{-1}\mathbf{u}_{t+h-l},
\end{align}
where $\mathbf{c}_{ih}'(\bm{\phi}) = \mathbf{e}_{i}'\mathbf{C}_{h}\bm{\Sigma}_{tr}$ is row $i$ of $\mathbf{C}_{h}\bm{\Sigma}_{tr}$ and $\mathbf{q}_{j}=\mathbf{Q}\mathbf{e}_{j}$ is column $j$ of $\mathbf{Q}$. For example, $H_{1,j,t,t}$ represents the contribution of shock $j$ to the one-step-ahead forecast error in $p_t$. 

An alternative definition is the contribution of \emph{all} past shock realisations to $y_{it}$ (\citealt{Kilian_Lutkepohl_2017,Plagborg-Moller_Wolf_2022,Bergholt_etal_2024}). In terms of (\ref{eq:histdecomp}), this definition corresponds to $H_{i,j,1,t}$. The difference between $y_{it}$ and $\sum_{j} H_{i,j,1,t}$ represents the contributions of initial conditions and deterministic terms (e.g. a constant). 

It is straightforward to show that $\sum_{j}H_{i,j,t,t+h} = \sum_{l=0}^{h}\mathbf{e}_{i}'\mathbf{C}_{l}\mathbf{u}_{t+h-l}$, which is the $(h+1)$-step-ahead forecast error in variable $i$. In the bivariate model, knowing the contribution of one shock means that we also know the contribution of the other shock. In what follows, I therefore focus on the contribution of $\varepsilon_{1t}$ to the forecast error in $p_t$ as the structural object of interest, and denote this by $H_{t,t+h} \equiv H_{1,1,t,t+h}$.\footnote{The analysis focuses on scalar historical decompositions, with inference conducted marginally for each shock contribution and period. In some applications, joint inference about contributions across periods or shocks may be of interest (e.g. \citealt{Inoue_Kilian_2022}). However, the associated joint identified sets may be high-dimensional and difficult to visualise and summarise. As discussed in \citet{Giacomini_Kitagawa_Read_2022_rejoinder}, posterior lower and upper probabilities provide one way to assess joint hypotheses without characterising the full joint identified set; Online Appendix~\ref{app:jointinference} applies this approach.}

\section{Identification of Historical Decompositions}
\label{sec:informativeness}

This section characterises identified sets for the historical decomposition, and explains how the identifying power of the restrictions depends on features of the data.

\subsection{Sign restrictions and identified sets}

Let $\mathrm{vech}(\bm{\Sigma}_{tr})= (\sigma_{11},\sigma_{21},\sigma_{22})'$. The correlation between the innovations in $p_t$ and $q_t$ is
 \begin{equation}
     \rho = \frac{\sigma_{21}}{\sqrt{\sigma_{21}^{2}+\sigma_{22}^{2}}}.
 \end{equation}
The space of $2\times 2$~orthonormal rotation matrices can be represented as
\begin{equation}
    \mathcal{O}(2) = \left\{
    \begin{bmatrix}
        \cos\theta & -\sin\theta \\
        \sin\theta & \cos\theta 
    \end{bmatrix}
    : \theta \in [-\pi,\pi]\right\},
\end{equation}
so the matrix of impact impulse responses is
\begin{equation}
    \mathbf{A}_{0}^{-1} =
    \begin{bmatrix}
        \sigma_{11}\cos\theta & -\sigma_{11}\sin\theta \\
        \sigma_{21}\cos\theta + \sigma_{22}\sin\theta & \sigma_{22}\cos\theta-\sigma_{21}\sin\theta
    \end{bmatrix}
.
\end{equation}
Impose the sign restrictions:
\begin{equation}
    \mathbf{A}_{0}^{-1} =
    \begin{bmatrix}
        + & + \\
        - & +
    \end{bmatrix},
\end{equation}
which require $p_t$ and $q_t$ to move in opposite directions following a supply shock, $\varepsilon_{st} \equiv \varepsilon_{1t}$, and in the same direction following a demand shock, $\varepsilon_{dt} \equiv \varepsilon_{2t}$. These sign restrictions generate an identified set for $\theta$:\footnote{For a derivation of this identified set, see \cite{Baumeister_Hamilton_2015} or \cite{Read_2022}.}
\begin{equation}\label{eq:istheta}
IS_{\theta}(\bm{\phi}) = 
\begin{cases}
\left[\arctan\left(\frac{\sigma_{22}}{\sigma_{21}}\right),0\right] & \text{if $\rho < 0$} \\
\left[-\frac{\pi}{2},\arctan\left(-\frac{\sigma_{21}}{\sigma_{22}}\right)\right] & \text{if $\rho \geq 0$.}
\end{cases}
\end{equation}
An identical identified set is obtained by imposing the sign restrictions
\begin{equation}
    \mathbf{A}_{0} =
    \begin{bmatrix}
        + & - \\
        + & +
    \end{bmatrix},
\end{equation}
which imply that the first structural equation can be interpreted as an upward-sloping supply curve and the second as a downward-sloping demand curve. The innovation correlation $\rho$ determines the identifying power of the sign restrictions for the slopes of each curve. Letting $\omega^{s}=-(\mathbf{e}_{1}'\mathbf{A}_0\mathbf{e}_2)/(\mathbf{e}_1'\mathbf{A}_0\mathbf{e}_1)$ be the slope of the supply curve and $\omega^{d}=-(\mathbf{e}_{2}'\mathbf{A}_0\mathbf{e}_2)/(\mathbf{e}_2'\mathbf{A}_0\mathbf{e}_1)$ be the slope of the demand curve, the identified sets for the slopes of the curves are:
\begin{equation}
    IS_{\omega^s}(\bm{\phi}) = 
\begin{cases}
    [0,\infty) & \text{if $\rho \leq  0$} \\
    \left[\frac{\sigma_{11}}{\sqrt{\sigma_{21}^{2}+\sigma_{22}^{2}}} \rho, \frac{\sigma_{11}}{\sqrt{\sigma_{21}^{2}+\sigma_{22}^{2}}} \frac{1}{\rho}\right] &\text{if $\rho >  0$,}
\end{cases}
\end{equation}
and 
\begin{equation}
    IS_{\omega^d}(\bm{\phi}) = 
\begin{cases}
    \left[\frac{\sigma_{11}}{\sqrt{\sigma_{21}^{2}+\sigma_{22}^{2}}}\frac{1}{\rho},\frac{\sigma_{11}}{\sqrt{\sigma_{21}^{2}+\sigma_{22}^{2}}}\rho\right] & \text{if $\rho < 0$} \\
    (-\infty,0] & \text{if $\rho \geq  0$.}
\end{cases}
\end{equation}
A strong negative correlation means that the sign restrictions tightly identify the slope of the demand curve, but leave the slope of the supply curve completely unidentified, and vice versa for a strong positive correlation. For example, as $\rho \rightarrow -1$, the forecast errors approximately trace out the demand curve but remain consistent with supply curves ranging from horizontal to vertical, as illustrated in  panel~(a) of Figure~\ref{fig:illustration}. These familiar results (e.g. \citealt{Leamer_1981,Uhlig_2017}) are useful in understanding the conditions under which the sign restrictions are informative about historical decompositions.\footnote{\cite{Giannone_Primiceri_2024_nber} use supply and demand curves to interpret historical decompositions of inflation, whereas I use them to understand the conditions under which sign restrictions are informative about these decompositions.}

\begin{figure}[h]
	\centering
	\caption{Identifying Historical Decompositions from Forecast Errors}\label{fig:illustration}
    \begin{tabular}{c c}
    \multicolumn{2}{c}{Panel~(a)} \\
    \includegraphics[scale=0.33]{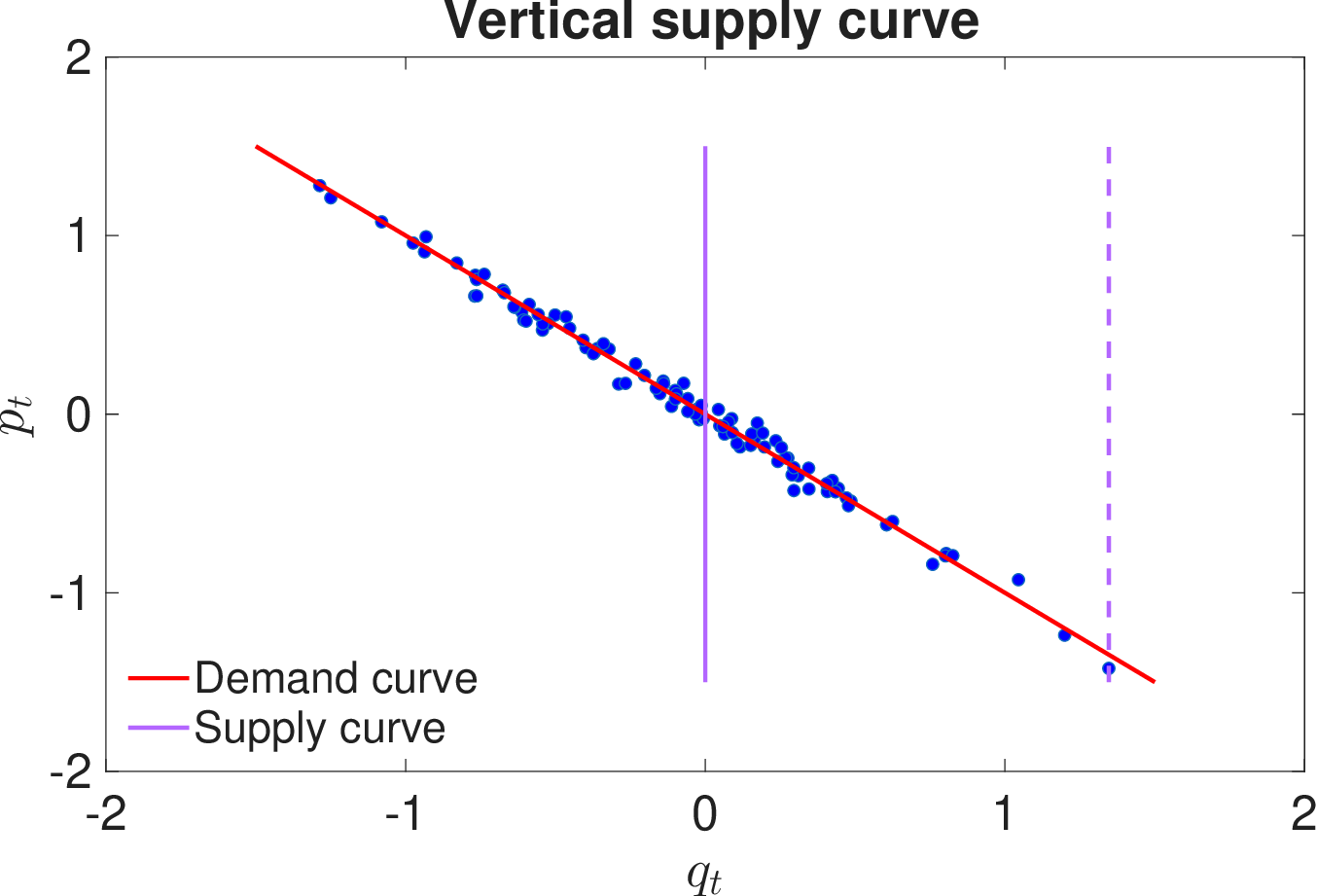} & \includegraphics[scale=0.33]{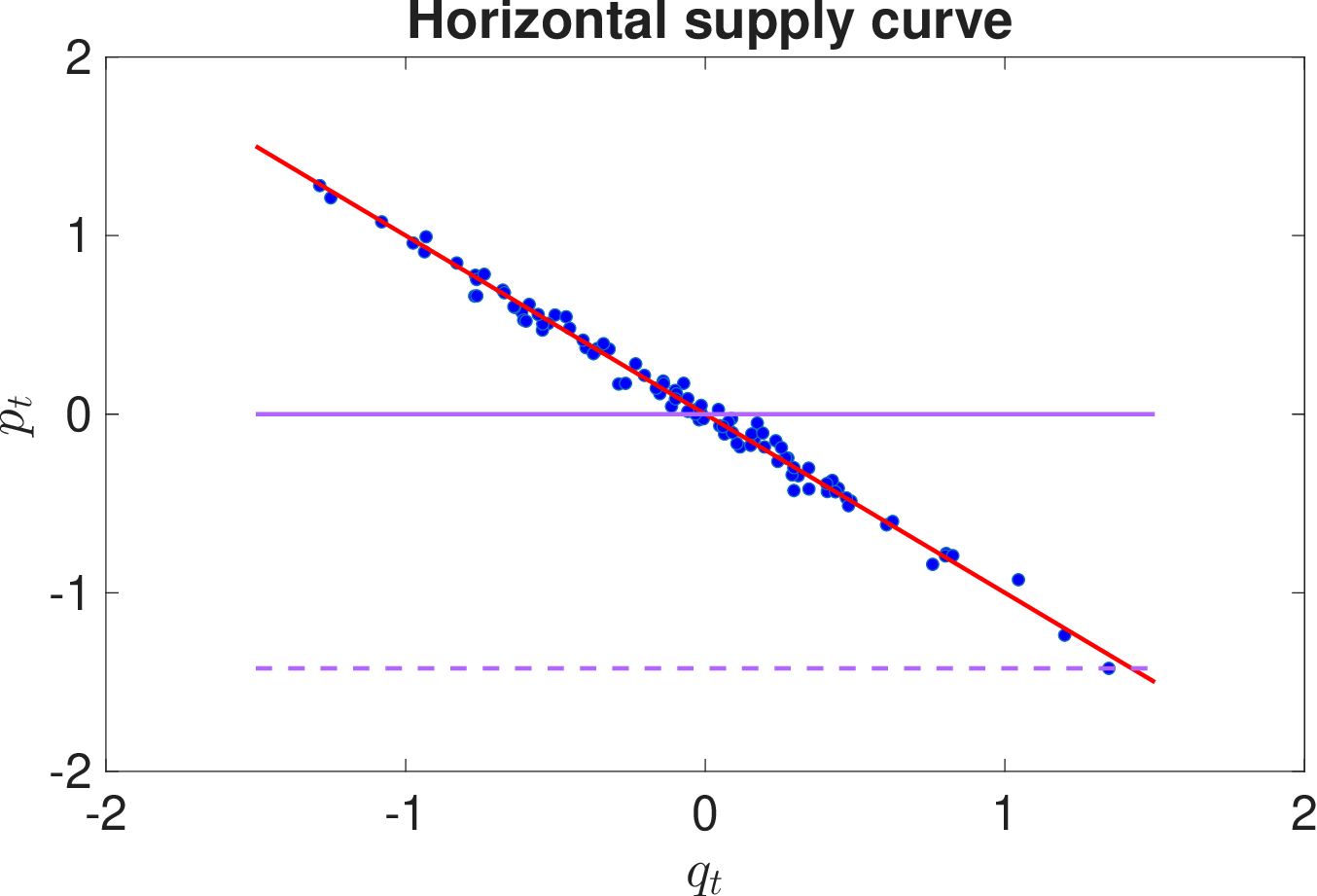} \\
    \multicolumn{2}{c}{Panel~(b)} \\
    \includegraphics[scale=0.33]{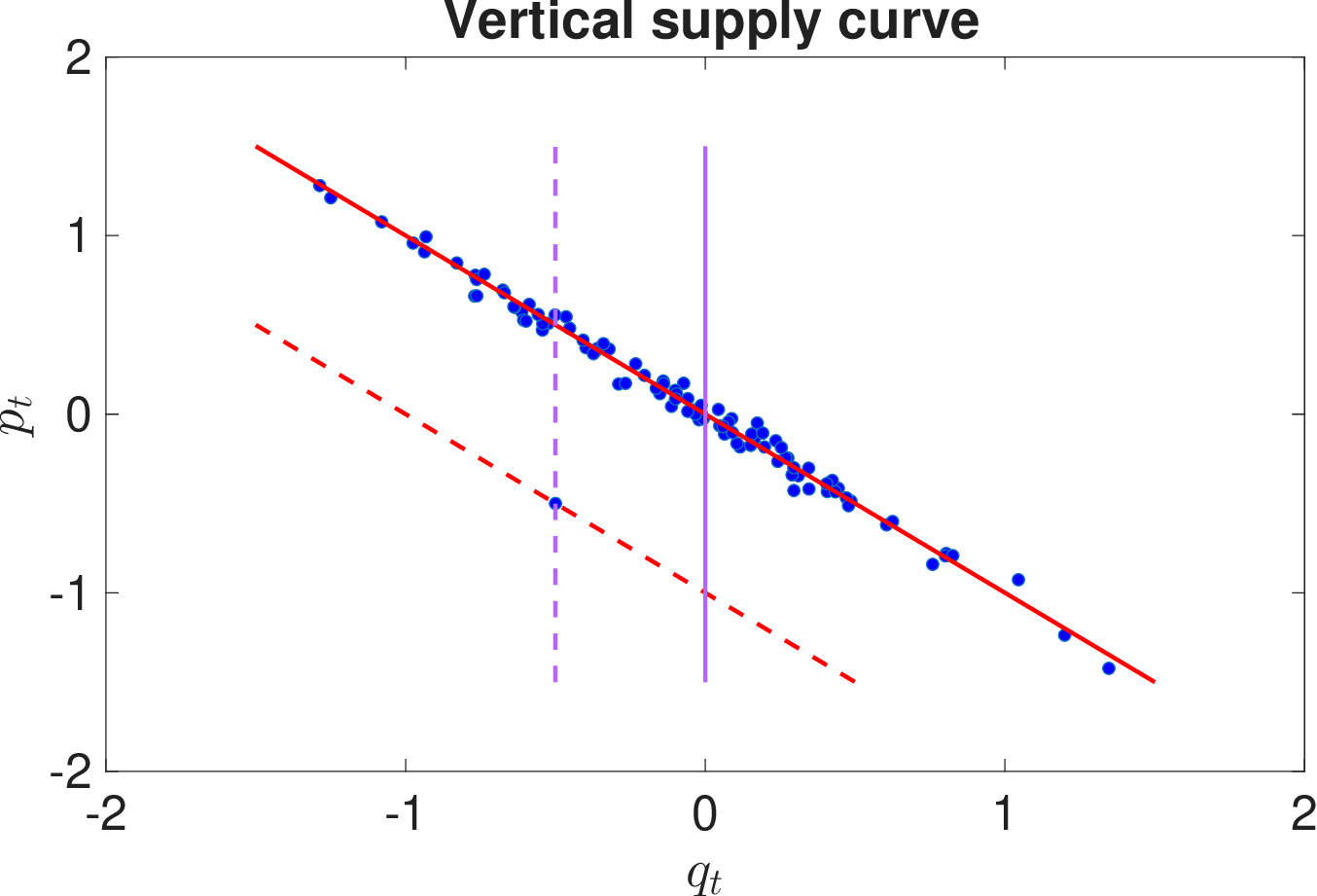} & \includegraphics[scale=0.33]{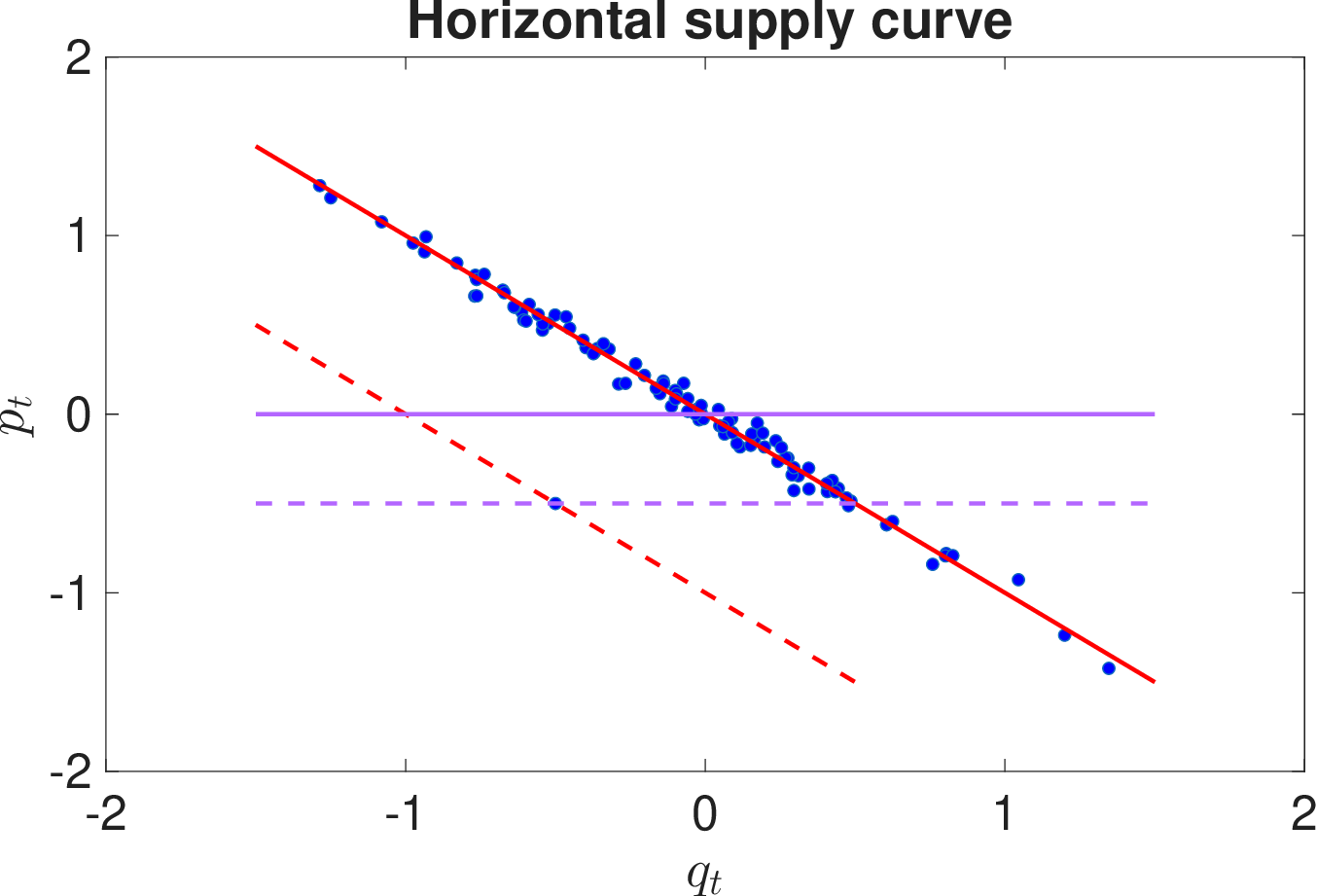}     
    \end{tabular}
    \noindent\begin{minipage}{\textwidth}
	\small\textbf{Notes:} Forecast errors simulated from data-generating process with $\rho \approx -1$. Panel~(a) illustrates possible decompositions of a forecast error that conforms to the historical relationship. Panel~(b) illustrates possible decompositions of a forecast error that departs from the historical relationship.
     \end{minipage}  
\end{figure}

\subsection{Historical decomposition}

In this two-variable model, we can write the historical decomposition of $p_t$ with respect to the supply shock, $H_{t,t+h}$, as a function of $\theta$, $\bm{\phi}$ and $\left\{\mathbf{u}_{l}\right\}_{l=t}^{t+h}$:
\begin{equation}\label{eq:hdtheta}
    H_{t,t+h}(\theta,\bm{\phi},\left\{\mathbf{u}_{l}\right\}_{l=t}^{t+h}) =  \begin{bmatrix}
        \cos\theta & \sin\theta
    \end{bmatrix}
    \bm{\Omega}_{h}(\bm{\phi}, \left\{\mathbf{u}_{l}\right\}_{l=t}^{t+h})
    \begin{bmatrix}
        \cos\theta \\
        \sin\theta
    \end{bmatrix},
\end{equation}
where 
\begin{equation}\label{eq:omega}
    \bm{\Omega}_{h}(\bm{\phi}, \left\{\mathbf{u}_{l}\right\}_{l=t}^{t+h}) = \sum_{l=0}^{h}\mathbf{c}_{1l}(\bm{\phi})(\bm{\Sigma}_{tr}^{-1}\mathbf{u}_{t+h-l})'.
\end{equation}
The identified set for $\theta$ in (\ref{eq:istheta}) can be viewed as inducing an identified set for $H_{t,t+h}$:\footnote{Since $IS_{\theta}(\bm{\phi})$ is an interval and $H_{t,t+h}(\theta,\bm{\phi},\left\{\mathbf{u}_{l}\right\}_{l=t}^{t+h})$ is a continuous function of $\theta$, $IS_{H_{t,t+h}}(\bm{\phi},\left\{\mathbf{u}_{l}\right\}_{l=t}^{t+h})$ is also an interval.}
\begin{align}
IS_{H_{t,t+h}}(\bm{\phi},\left\{\mathbf{u}_{l}\right\}_{l=t}^{t+h}) &= \left\{H_{t,t+h}(\theta,\bm{\phi},\left\{\mathbf{u}_{l}\right\}_{l=t}^{t+h}): \theta \in IS_{\theta}(\bm{\phi})\right\} \\
&= \left[L_{t,t+h}(\bm{\phi},\left\{\mathbf{u}_{l}\right\}_{l=t}^{t+h}),U_{t,t+h}(\bm{\phi},\left\{\mathbf{u}_{l}\right\}_{l=t}^{t+h})\right],
\end{align}
where (suppressing arguments) $L_{t,t+h}$ and $U_{t,t+h}$ are, respectively, the minimum and maximum of $H_{t,t+h}$ over $IS_{\theta}(\bm{\phi})$. Any value of $H_{t,t+h}$ in $IS_{H_{t,t+h}}(\bm{\phi},\left\{\mathbf{u}_{l}\right\}_{l=t}^{t+h})$ is equally consistent with the identifying restrictions, the second moments of the data and the observed forecast errors. Online Appendix~\ref{sec:computation} describes a computationally simple and accurate numerical procedure for computing these identified sets in practice.

\subsubsection{When are the sign restrictions informative?}

Assume the object of interest is the contribution of the supply shock to the one-step-ahead forecast error in $p_t$:
\begin{equation}\label{eq:hd_bivariate}
    H_{t,t}(\theta,\bm{\phi},\mathbf{u}_{t}) =  u_{pt}\cos^{2}\theta + \frac{1}{\sigma_{22}}(\sigma_{11}u_{qt}-\sigma_{21}u_{pt})\sin\theta\cos\theta.
\end{equation}
This case is useful to consider because the identified set is analytically tractable, making it clear how features of the data govern the strength of identification. The following discussion draws on an analytical characterisation of $IS_{H_{t,t}}(\bm{\phi},\left\{\mathbf{u}_{l}\right\}_{l=t}^{t+h})$, presented in Online Appendix~\ref{sec:analytical}.\footnote{The discussion here would also apply to other concepts of the historical decomposition (e.g. $H_{1,t}$) in cases where shocks have non-zero effects only on impact; in that case, $\mathbf{c}_{1l}'(\bm{\phi}) = (0, 0)$ for all $l > 0$ and (\ref{eq:hd_bivariate}) is equivalent to $H_{s,t}(\theta,\bm{\phi},\{\mathbf{u}_{l}\}_{l=1}^{t})$ for $1 \leq s \leq t$.}  

The identified set always admits one `extreme' decomposition. If $\rho < 0$, then $0 \in IS_{\theta}(\bm{\phi})$ and evaluating (\ref{eq:hd_bivariate}) at $\theta = 0$ gives $H_{t,t} = u_{pt}$. The price forecast error may therefore be attributed entirely to the supply shock, corresponding to a horizontal supply curve. If $\rho \geq 0$, then $-\pi/2 \in IS_{\theta}(\bm{\phi})$, and evaluating (\ref{eq:hd_bivariate}) at $\theta = -\pi/2$ gives $H_{t,t} = 0$. The price forecast error may then be attributed entirely to the demand shock, corresponding to a horizontal demand curve. It follows that the restrictions are uninformative when the reduced-form innovations are uncorrelated. If $\rho = 0$, then $IS_{\theta}(\bm{\phi}) = [-\pi/2,0]$, so both zero and $u_{pt}$ are in $IS_{H_{t,t}}(\bm{\phi},\mathbf{u}_t)$. Thus, the observed price forecast error can be explained entirely by either a demand or a supply shock. By continuity, the restrictions will be relatively uninformative when $\rho$ is close to zero.

A large absolute correlation can make the restrictions more informative, but only if the realised forecast error aligns closely with the historical price-quantity relationship. Suppose, for example, that $\rho \approx -1$, so the forecast errors approximately trace out a downward-sloping demand curve. Although supply curves with widely differing slopes remain admissible, a forecast error lying close to this historical relationship will imply similar supply contributions across those curves. Panel~(a) of Figure~\ref{fig:illustration} illustrates this case; whether the admissible supply curve is vertical, horizontal or has an intermediate slope, the observed price change is attributed predominantly to the supply shock, yielding a narrow identified set.

Strong correlation is not sufficient, however, for the sign restrictions to deliver informative identified sets. If the realised forecast error departs from the historical relationship, its decomposition may vary substantially across admissible models. In panel~(b), the same decline in price and quantity can be generated under both vertical and horizontal supply curves, but the implied supply contributions have opposite signs. The identified set is therefore wide despite the fact that $\rho \approx -1$. Hence, $\rho$ governs only the \emph{potential} identifying power of the sign restrictions, while the realised forecast error determines whether that potential is realised. To summarise, when interest is in the historical decomposition, $\rho$ is not a sufficient statistic for the identifying power of the restrictions. This contrasts with other structural objects, such as impulse responses, elasticities and forecast error variance decompositions (e.g. \citealt{Baumeister_Hamilton_2015,Uhlig_2017,Inoue_Kilian_2026}).

\subsubsection{Beyond the bivariate model}
\label{sec:beyondbivariate}

The discussion above has focused on the bivariate SVAR, motivated by its use in empirical settings, but a natural question is whether the features of the identification problem extend to larger SVARs. To examine this question, Online Appendix~\ref{app:beyondbivariate} considers a three-variable model in which the supply-demand sign restrictions are imposed on prices and quantities, while the responses associated with the additional variable and shock are left unrestricted. 

The example shows that adding a variable without imposing additional restrictions need not sharpen identification and may instead make identified sets for the historical decomposition substantially less informative. Even when price and quantity innovations are strongly correlated and the realised forecast errors conform closely to their historical relationship, the identified set for the supply contribution can remain wide; in the language of \cite{Wolf_2020}, the additional unidentified shock can `masquerade' as the supply shock.

Thus, the favourable conditions under which supply-demand decompositions can be tightly identified in the bivariate model are not sufficient in larger systems. Adding variables is unlikely to sharpen identification in the absence of credible identifying restrictions that distinguish the additional shocks from the supply and demand shocks of interest.

\subsubsection{The role of additional identifying restrictions}

When the supply-demand sign restrictions yield wide identified sets, imposing additional restrictions could sharpen identification and may be useful in particular applications. However, a crucial question is whether they can, in general, be credibly applied in this setting. 

Two types of restrictions that have been used to sharpen identification in other settings are narrative restrictions (e.g. \citealt{Antolin-Diaz_Rubio-Ramirez_2018,Ludvigson_Ma_Ng_2021}) and restrictions based on external instruments (e.g. \citealt{Ludvigson_Ma_Ng_2021,Braun_Bruggemann_2023}). A challenge in exploiting those restrictions in the current setting is that the disturbances in the bivariate model need not correspond to particular fundamental shocks; if the data-generating process contains more structural shocks than variables in the VAR, the VAR disturbances may combine current and past shocks of different economic types (\citealt{Canova_Ferroni_2022}). Consequently, a narrative or instrument relating to a specific shock, such as a monetary policy or oil supply shock, cannot necessarily be mapped to a restriction on the model's supply or demand disturbances without additional assumptions. 

Restrictions on dynamic responses or elasticities (e.g. \citealt{Kilian_Murphy_2012}) could also sharpen identification, but their credibility would depend on information beyond the basic assumptions that supply curves slope upwards and demand curves slope downward. Thus, additional restrictions may be useful in particular applications, but they necessarily introduce substantive assumptions beyond those whose identifying content is examined here.  

\section{Supply-demand Decompositions of US Inflation}
\label{sec:empirical}

This section empirically assesses the identifying power of the supply-demand sign restrictions using aggregate data, following \cite{Giannone_Primiceri_2024_nber} (henceforth, GP24), and disaggregated data, motivated by \cite{Shapiro_2026}.

\subsection{Aggregate decomposition}
\label{subsec:agg}

Using a two-variable SVAR identified via sign restrictions on the impulse responses to supply and demand shocks, GP24 decompose US inflation into contributions from shocks to aggregate supply and demand. GP24 estimate a four-lag VAR in log CPI and log real GDP over 1997:Q1--2019:Q4 and decompose post-2019 forecast errors in year-ended inflation.

I estimate the same reduced-form VAR as GP24, using their specification, sample and Bayesian prior for the reduced-form parameters.\footnote{I am grateful to Domenico Giannone and Giorgio Primiceri for providing me with their replication code.} The reduced-form prior combines Minnesota and sum-of-coefficients priors, following \cite{Giannone_Lenza_Primiceri_2015}.\footnote{The sum-of-coefficients prior reduces posterior uncertainty about the deterministic component, which can be an important source of uncertainty about historical decompositions (\citealt{Bergholt_etal_2024}).} GP24 supplement the sign restrictions with a uniform prior for $\mathbf{Q}$ (equivalently, $\theta$) and report pointwise posterior means of the historical decompositions. When estimating identified sets, I use the prior-robust Bayesian approach of \cite{Giacomini_Kitagawa_2021}, which considers the class of conditional priors assigning probability one to $IS_{\theta}(\bm{\phi})$. I summarise the resulting class of posteriors using the set of posterior means and a robust credible interval. The former provides an estimate of the identified set, while the latter captures posterior uncertainty about the reduced-form parameters. Online Appendix~\ref{app:aggregate_imp} provides implementation details.

\begin{figure}[!htbp]
	\centering
	\caption{Contribution of Supply Shocks to Year-ended US CPI Inflation}\label{fig:GP24}
         \includegraphics[scale=0.6]{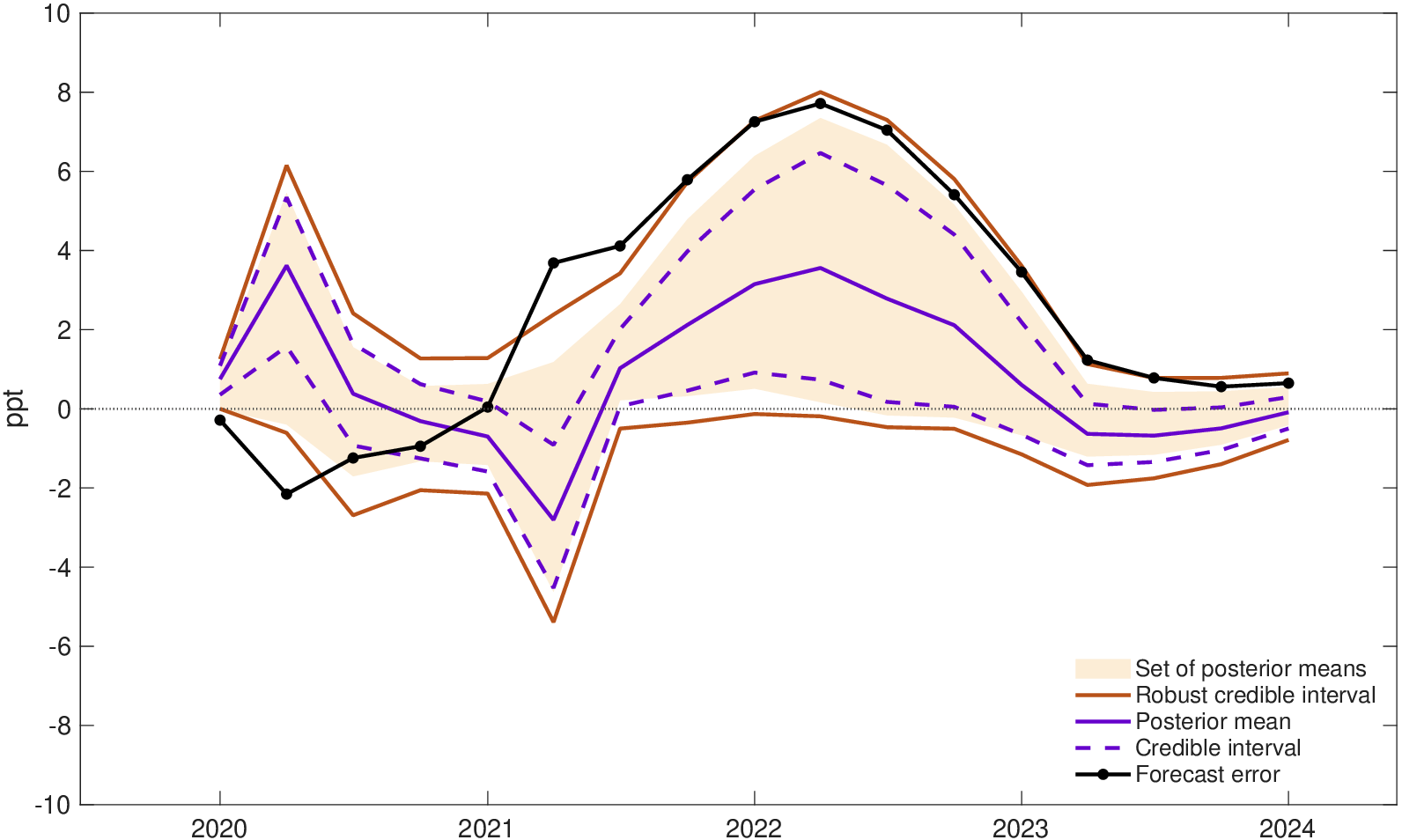}
    \noindent\begin{minipage}{\textwidth}
	\small\textbf{Notes:} Contribution of supply shocks to forecast error in year-ended CPI inflation given information available up to 2019:Q4. Set of posterior means is interpretable as estimator of identified set. Standard and robust credible intervals are at the 68~per cent credibility level.
     \end{minipage}  
\end{figure}

Figure~\ref{fig:GP24} plots the supply contributions under both the uniform prior and the robust Bayesian approach. The posterior mean contributions under the uniform prior replicate the results in GP24; at its peak in 2022:Q2, a bit less than half of the 7.7~percentage point forecast error in year-ended inflation was attributed to supply shocks (though there is substantial posterior uncertainty about this contribution). In contrast, the set of posterior means spans 0.2--7.3~percentage points, so the restrictions do not determine whether the inflation forecast error was predominantly supply- or demand-driven. This uninformativeness partly reflects a relatively weak correlation between price and quantity innovations over the estimation period; the posterior mean for $\rho$ is only 0.23 with 68~per cent credible interval $[0.13,0.33]$.\footnote{In an analogous exercise, GP24 also decompose consumer price inflation in the euro area. Online Appendix~\ref{app:euro} repeats the analysis above in this setting; the sign restrictions are again largely uninformative about the magnitude of the supply and demand contributions. Online Appendix D.3 extends the analysis of US inflation to the three-variable SVARs considered in GP24. The resulting identified sets remain wide, indicating that augmenting the model with an additional variable and additional sign restrictions does not materially sharpen identification in these particular cases.}

\subsection{Disaggregated decompositions}
\label{subsec:disagg}

To assess whether disaggregation sharpens identification, I follow \cite{Shapiro_2026} and estimate decompositions for 136 categories in the Personal Consumption Expenditures (PCE) basket. For each category $k$, I estimate by OLS a 12-lag VAR in monthly PCE inflation and real expenditure growth over 1988:M1--2023:M9. Online Appendix~\ref{app:disaggregated_det} provides details. 

\begin{figure}[htp]
	\centering
	\caption{Supply Shock Contribution to Inflation in Selected PCE Categories}
	\label{fig:disaggregated_decomp}
    \begin{tabular}{c}
          \includegraphics[scale=0.6]{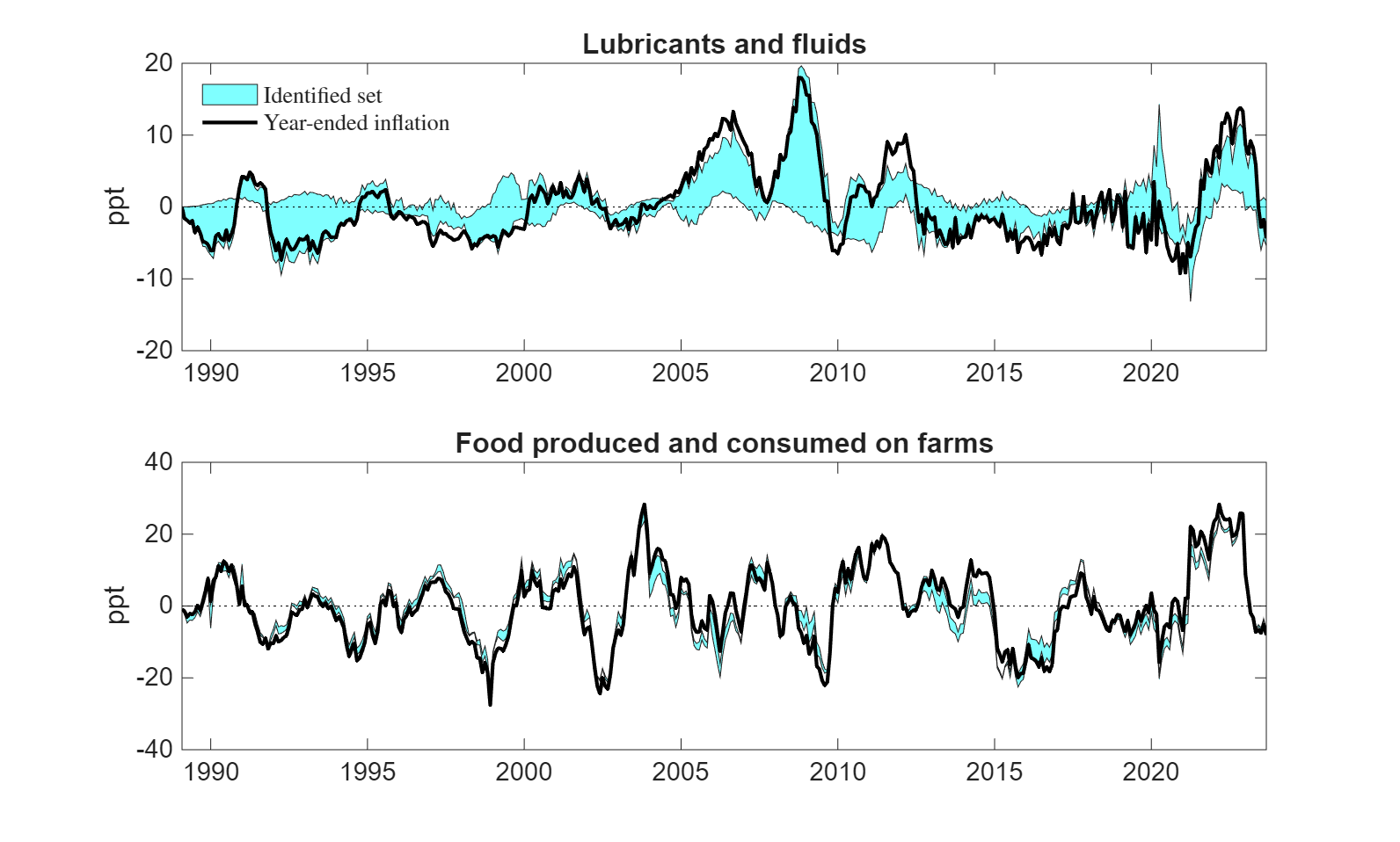}
    \end{tabular}
    \noindent\begin{minipage}{\textwidth}
	\small\textbf{Notes:} Categories with smallest and largest absolute innovation correlations. Inflation is in deviations from deterministic component (i.e. contribution of constant and initial conditions).
    \end{minipage}
\end{figure}

At this level of disaggregation, price-quantity correlations tend to be weak; absolute correlations exceed 0.2 in only about half of categories and are below 0.4 in about 80~per cent. Identified sets for historical decompositions therefore tend to be wide. Figure~\ref{fig:disaggregated_decomp} illustrates this result using the categories with the smallest and largest absolute correlations. The category with the smallest absolute correlation is `lubricants and fluids' ($|\rho| < 0.0001$); for this category, the identified set often includes both zero and the observed inflation rate, implying that the sign restrictions cannot distinguish between supply- and demand-driven price changes. In contrast, for `food produced and consumed on farms', $\rho = -0.98$ and the identified sets are narrow in most periods, though their width varies over time depending on the realisations of the innovations. The strongly negatively correlated innovations approximately trace out a short-run demand curve, allowing inflation in this category to be attributed predominantly to supply shocks. For example, in March 2022, the estimated supply contribution was 23.9--25.1 percentage points, compared with inflation of 28.4 per cent (net of deterministic components). Thus, disaggregation occasionally yields informative decompositions, but the sign restrictions remain weakly informative for most categories.

\section{Conclusion}
\label{sec:conclusion}

Sign restrictions on the slopes of supply and demand curves are appealing because they impose relatively uncontroversial assumptions. However, they generally set-identify historical decompositions and need not determine whether observed price changes were predominantly supply- or demand-driven. The identifying power of the restrictions with respect to historical decompositions depends both on reduced-form parameters and realised forecast errors, unlike impulse responses, elasticities and forecast error variance decompositions, where the strength of identification depends on the reduced-form parameters alone. 

Empirically, the restrictions yield wide identified sets for decompositions of aggregate US inflation and in most expenditure categories, although disaggregation produces sharp decompositions in a few cases with unusually strong correlations between price and quantity innovations. Consequently, conclusions  about the drivers of inflation based on a single admissible decomposition may be driven largely by auxiliary assumptions, such as the prior used to select among observationally equivalent models.

\putbib
\end{bibunit}

\clearpage

\makeatletter
\gdef\@thanks{}
\renewcommand{\thefootnote}{\fnsymbol{footnote}} 
\setcounter{footnote}{1}                          
\makeatother

\title{Supplemental Appendix for ``Sign Restrictions and Supply-Demand Decompositions of Inflation''}
\author{Matthew Read\thanks{Economic Research Department, Reserve Bank of Australia. email: readm@rba.gov.au}}
\date{\today}    

\begin{bibunit}
\maketitle

\renewcommand{\thefootnote}{\arabic{footnote}}
\setcounter{footnote}{0}

This Supplemental Appendix contains additional material supporting the analysis in the main text. Appendix~\ref{sec:computation} explains how to numerically compute the end points of the identified set for the historical decomposition. Appendix~\ref{sec:analytical} analytically characterises identified sets for historical decompositions in the special case where interest is in decomposing one-step-ahead forecast errors. Appendix~\ref{app:beyondbivariate} examines the informativeness of the supply-demand sign restrictions when adding an additional variable to the model. Appendix~\ref{app:additionalresults} provides additional details related to the empirical applications.

\renewcommand{\theequation}{A\arabic{equation}}
\renewcommand{\thesection}{A}
\renewcommand{\thefigure}{A\arabic{figure}}
\setcounter{equation}{0}
\setcounter{figure}{0}
\section{Computing identified sets}
\label{sec:computation}

This appendix explains how to numerically compute the end points of the identified set for the historical decomposition, $H_{t,t+h}$.

$IS_{H_{t,t+h}}(\bm{\phi},\left\{\mathbf{u}_{l}\right\}_{l=t}^{t+h})$ can be computed by solving the optimisation problems that define $L_{t,t+h}(\bm{\phi},\left\{\mathbf{u}_{l}\right\}_{l=t}^{t+h})$ and $U_{t,t+h}(\bm{\phi},\left\{\mathbf{u}_{l}\right\}_{l=t}^{t+h})$. Focusing on the upper bound, $U_{t,t+h}(\bm{\phi},\left\{\mathbf{u}_{l}\right\}_{l=t}^{t+h})$ corresponds to either an end point of $IS_{\theta}(\bm{\phi})$ or a critical point of $H_{t,t+h}(\theta,\bm{\phi},\left\{\mathbf{u}_{l}\right\}_{l=t}^{t+h})$ in the interior of $IS_{\theta}(\bm{\phi})$. Evaluating $H_{t,t+h}(\theta,\bm{\phi},\left\{\mathbf{u}_{l}\right\}_{l=t}^{t+h})$ at the end points of $IS_{\theta}(\bm{\phi})$ is straightforward using (\ref{eq:istheta}) and (\ref{eq:hdtheta}). To find a critical point, rewrite $H_{t,t+h}$ as a function of $\mathbf{q}_{1}$. The problem $\max_{\mathbf{q}_{1}}H_{t,t+h}(\mathbf{q}_{1},\bm{\phi},\left\{\mathbf{u}_{l}\right\}_{l=t}^{t+h})$ subject to $\mathbf{q}_{1}'\mathbf{q}_{1} = 1$ has first-order necessary condition $\mathbf{S}_h(\bm{\phi},\left\{\mathbf{u}_{l}\right\}_{l=t}^{t+h})\mathbf{q}_1 =  \lambda \mathbf{q}_{1}$, where $\mathbf{S}_h(\bm{\phi},\left\{\mathbf{u}_{l}\right\}_{l=t}^{t+h}) = (1/2)(\bm{\Omega}_{h}(\bm{\phi}, \left\{\mathbf{u}_{l}\right\}_{l=t}^{t+h})+\bm{\Omega}_{h}(\bm{\phi}, \left\{\mathbf{u}_{l}\right\}_{l=t}^{t+h})')$ and $\lambda$ is the Lagrange multiplier on the constraint. Solutions to this equation are eigenvectors of $\mathbf{S}_h(\bm{\phi},\left\{\mathbf{u}_{l}\right\}_{l=t}^{t+h})$. I compute the two eigenvectors, normalise their signs so that the sign normalisation $\mathbf{e}_{1}'\mathbf{A}_{0}\mathbf{e}_{1} = (\bm{\Sigma}_{tr}^{-1}\mathbf{e}_{1})'\mathbf{q}_{1} \geq 0$ holds, and check whether either normalised eigenvector lies within the identified set for $\mathbf{q}_1$.\footnote{The identified set for $\mathbf{q}_1$ is $\{\mathbf{q}_1 = (\cos\theta, \sin\theta)': \theta \in IS_{\theta}(\bm{\phi})\}$. Checking whether $\mathbf{q}_1^{*} = (q_{1,1}^*,q_{1,2}^*)'$ lies in the identified set for $\mathbf{q}_1$ is equivalent to checking whether $\theta^{*} = \operatorname{atan2}(q_{1,2}^*,q_{1,1}^*)$ lies in $IS_{\theta}(\bm{\phi})$, where $\operatorname{atan2}$ denotes the two-argument arctangent, taking values in $(-\pi,\pi]$.} $U_{t,t+h}(\bm{\phi},\left\{\mathbf{u}_{l}\right\}_{l=t}^{t+h})$ is then obtained by direct comparison of the function values at the end points and admissible critical points (if any). $L_{t,t+h}(\bm{\phi},\left\{\mathbf{u}_{l}\right\}_{l=t}^{t+h})$ is obtained similarly.

Often the historical decomposition of interest involves a transformation of the variables entering the VAR. In cases like this, the numerical procedure above can easily be adapted to compute the end points of the identified set. To give an example, the VAR from \cite{Giannone_Primiceri_2024_nber} (examined in Section~\ref{subsec:agg}) includes the price level $p_t$ but we are interested in decomposing year-ended inflation. The data are quarterly, so year-ended inflation is $\pi_t^{(ye)} =p_t - p_{t-4}$. Given a particular historical decomposition, $H_{t,t+h}$, we can obtain the historical decomposition of $\pi_{t+h}^{(ye)}$ as $H_{t,t+h}$ for $h=0,\ldots,3$ and $H_{t,t+h}-H_{t,t+h-4}$ for $h\geq 4$. From (\ref{eq:hdtheta}), it follows that, for $h\geq 4$,
\begin{equation}\label{eq:hdtheta_ye}
    H_{t,t+h}(\theta,\bm{\phi},\left\{\mathbf{u}_{l}\right\}_{l=t}^{t+h}) - H_{t,t+h-4}(\theta,\bm{\phi},\left\{\mathbf{u}_{l}\right\}_{l=t}^{t+h-4}) =  \begin{bmatrix}
        \cos\theta & \sin\theta
    \end{bmatrix}
    \tilde{\bm{\Omega}}_{h}(\bm{\phi}, \left\{\mathbf{u}_{l}\right\}_{l=t}^{t+h})
    \begin{bmatrix}
        \cos\theta \\
        \sin\theta
    \end{bmatrix},
\end{equation}
where
\begin{equation}
    \tilde{\bm{\Omega}}_{h}(\bm{\phi}, \left\{\mathbf{u}_{l}\right\}_{l=t}^{t+h}) = \bm{\Omega}_{h}(\bm{\phi}, \left\{\mathbf{u}_{l}\right\}_{l=t}^{t+h}) - \bm{\Omega}_{h-4}(\bm{\phi}, \left\{\mathbf{u}_{l}\right\}_{l=t}^{t+h-4}).
\end{equation}
The problem of finding the end points of the identified set for $H_{t,t+h}-H_{t,t+h-4}$ therefore has exactly the same structure as finding the end points of the identified set for $H_{t,t+h}$. Hence, the numerical procedure above can be applied with $\tilde{\bm{\Omega}}_{h}(\bm{\phi}, \left\{\mathbf{u}_{l}\right\}_{l=t}^{t+h})$ replacing $\bm{\Omega}_{h}(\bm{\phi}, \left\{\mathbf{u}_{l}\right\}_{l=t}^{t+h})$.\footnote{Similar reasoning applies when the endogenous variable entering the VAR is a monthly inflation rate, but the object of interest is the historical decomposition of year-ended inflation; see Online Appendix~\ref{app:disaggregated_det} for further discussion.}

An alternative procedure to compute the end points of the identified set is to obtain many random draws of $\theta$ from a uniform distribution over $IS_{\theta}(\bm{\phi})$, evaluate $H_{t,t+h}(\theta,\bm{\phi},\left\{\mathbf{u}_{l}\right\}_{l=t}^{t+h})$ at each draw and approximate the end points of $IS_{H_{t,t+h}}(\bm{\phi},\left\{\mathbf{u}_{l}\right\}_{l=t}^{t+h})$ using the minimum and maximum values of $H_{t,t+h}(\theta,\bm{\phi},\left\{\mathbf{u}_{l}\right\}_{l=t}^{t+h})$ over these draws. A drawback of this approach is that it suffers from approximation error that vanishes only as the number of draws of $\theta$ diverges (e.g. \citealt{Montiel-Olea_Nesbit_2021}). In contrast, the numerical procedure described above does not suffer from such approximation error. The procedure is also computationally simple, requiring evaluating $H_{t,t+h}(\theta,\bm{\phi},\left\{\mathbf{u}_{l}\right\}_{l=t}^{t+h})$ at a small number of candidate values (the two end points of $IS_{\theta}(\bm{\phi})$ and at most two admissible eigenvectors of $\mathbf{S}_h(\bm{\phi},\left\{\mathbf{u}_{l}\right\}_{l=t}^{t+h})$).

\renewcommand{\theequation}{B\arabic{equation}}
\renewcommand{\thesection}{B}
\renewcommand{\thefigure}{B\arabic{figure}}
\setcounter{equation}{0}
\setcounter{figure}{0}
\section{Analytical Characterisation of the Identified Set}
\label{sec:analytical}

This appendix analytically characterises the identified set for the contribution of the supply shock to the one-step-ahead forecast error in $p_t$, $H_{t,t}$. Let 
\begin{equation}
    C(\bm{\phi},\mathbf{u}_t) = \frac{(\sigma_{11}u_{qt} - \sigma_{21}u_{pt})}{\sigma_{22}}.
\end{equation}
From (\ref{eq:hd_bivariate}), $H_{t,t}$ can then be written as
\begin{equation}
	H_{t,t}(\theta,\bm{\phi},\mathbf{u}_t) = u_{pt}\cos^{2}\theta + C(\bm{\phi},\mathbf{u}_t)\sin\theta \cos \theta.
\end{equation}
Let 
\begin{equation}
    IS_{\theta}(\bm{\phi}) = [\underline{\theta}(\bm{\phi}),\overline{\theta}(\bm{\phi})] = \begin{cases}
\left[\arctan\left(\frac{\sigma_{22}}{\sigma_{21}}\right),0\right] & \text{if $\rho < 0$} \\
\left[-\frac{\pi}{2},\arctan\left(-\frac{\sigma_{21}}{\sigma_{22}}\right)\right] & \text{if $\rho \geq 0$.}
\end{cases}
\end{equation}
Given the sign restrictions, $IS_{\theta}(\bm{\phi})$ induces an identified set for $H_{t,t}$:
\begin{equation}
    IS_{H_{t,t}}(\bm{\phi},\mathbf{u}_t) = \{H_{t,t}(\theta,\bm{\phi},\mathbf{u}_t) : \theta \in IS_{\theta}(\bm{\phi})\}.
\end{equation}
Because $IS_{\theta}(\bm{\phi})$ is a closed interval and $H_{t,t}(\theta,\bm{\phi},\mathbf{u}_t)$ is continuous in $\theta$, $IS_{H_{t,t}}(\bm{\phi},\mathbf{u}_t)$ is a closed interval with end points attained either at the end points of $IS_{\theta}(\bm{\phi})$ or at an interior critical point of $H_{t,t}(\theta,\bm{\phi},\mathbf{u}_t)$. Proposition~\ref{prop:htt_analytical} formally characterises the end points of $IS_{H_{t,t}}(\bm{\phi},\mathbf{u}_t)$.

\begin{proposition}\label{prop:htt_analytical}
    Suppose $u_{pt} \neq 0$ and $C(\bm{\phi},\mathbf{u}_t) \neq 0$. Define
    \begin{equation}
        R(\bm{\phi},\mathbf{u}_t) = \frac{C(\bm{\phi},\mathbf{u}_t)}{u_{pt}}
    \end{equation}
    and
    \begin{equation}\label{eq:thetastar}
    	\theta^*(\bm{\phi},\mathbf{u}_t) = 
    	\begin{cases}
    		\frac{1}{2} \arctan\left(R(\bm{\phi},\mathbf{u}_t)\right) & \text{if $R(\bm{\phi},\mathbf{u}_t) < 0$} \\
    		\frac{1}{2} \arctan\left(R(\bm{\phi},\mathbf{u}_t)\right) - \frac{\pi}{2} & \text{if $R(\bm{\phi},\mathbf{u}_t) > 0$.} \\		
    	\end{cases}
    \end{equation}
    The only possible interior critical point is $\theta^*(\bm{\phi},\mathbf{u}_t)$ defined in (\ref{eq:thetastar}).

    Define the set
    \begin{equation}\label{eq:thetaset}
        \Theta(\bm{\phi},\mathbf{u}_t) = \{\underline{\theta}(\bm{\phi}),\overline{\theta}(\bm{\phi})\} \cup \{\theta^{*}(\bm{\phi},\mathbf{u}_t):\theta^{*}(\bm{\phi},\mathbf{u}_t) \in \operatorname{int}IS_{\theta}(\bm{\phi})\}.
    \end{equation}
    Then
    \begin{equation}\label{eq:is_htt}
        IS_{H_{t,t}}(\bm{\phi},\mathbf{u}_t) = \left[L_{t,t}(\bm{\phi},\mathbf{u}_t),U_{t,t}(\bm{\phi},\mathbf{u}_t)\right],
    \end{equation}
    where 
    \begin{equation}\label{eq:is_htt_bounds}
    L_{t,t}(\bm{\phi},\mathbf{u}_t) = \min_{\theta \in \Theta(\bm{\phi},\mathbf{u}_t)} H_{t,t}(\theta,\bm{\phi},\mathbf{u}_t) \quad \text{and} \quad  U_{t,t}(\bm{\phi},\mathbf{u}_t) = \max_{\theta \in \Theta(\bm{\phi},\mathbf{u}_t)} H_{t,t}(\theta,\bm{\phi},\mathbf{u}_t).
    \end{equation}
    The minimum and maximum of $H_{t,t}$ over $IS_{\theta}(\bm{\phi})$ are therefore attained over the candidate set $\Theta(\bm{\phi},\mathbf{u}_t)$.
\end{proposition}

\begin{proof}
    Since $IS_{\theta}(\bm{\phi})$ is a closed interval and $H_{t,t}(\theta,\bm{\phi},\mathbf{u}_t)$ is continuous in $\theta$, the minimum and maximum of $H_{t,t}$ over $IS_{\theta}(\bm{\phi})$ are attained at an end point of $IS_{\theta}(\bm{\phi})$ or at an interior critical point. A critical point of $H_{t,t}(\theta,\bm{\phi},\mathbf{u}_t)$ satisfies the first-order condition
    \begin{equation}\label{eq:foc}
    	\frac{\partial H_{t,t}(\theta,\bm{\phi},\mathbf{u}_t)}{\partial \theta} = -u_{pt}\sin(2\theta) + C(\bm{\phi},\mathbf{u}_t) \cos(2\theta) = 0.
    \end{equation}
    Except on the events $u_{pt} = 0$ and $C(\bm{\phi},\mathbf{u}_t) =0$, which have probability zero under an absolutely continuous distribution for $\mathbf{u}_t$, the first-order condition is equivalent to\footnote{The excepted cases can be handled directly from (\ref{eq:foc}). If $C(\bm{\phi},\mathbf{u}_t) = 0$ and $u_{pt} \neq 0$,  there is no interior critical point in $(-\pi/2,0)$. If $u_{pt} = 0$ and $C(\bm{\phi},\mathbf{u}_t) \neq 0$, the unique interior critical point occurs at $-\pi/4$.}
       \begin{equation}
    	\tan(2\theta^*(\bm{\phi},\mathbf{u}_t)) = R(\bm{\phi},\mathbf{u}_t).
        \end{equation}
    The unique solution in $(-\pi/2,0)$ is $\theta^*(\bm{\phi},\mathbf{u}_t)$ defined in (\ref{eq:thetastar}).\footnote{Solutions to $\tan(2\theta) = R(\bm{\phi},\mathbf{u}_t)$ differ by $\pi/2$. Consequently, $IS_{\theta}(\bm{\phi}) \subseteq [-\pi/2,0]$ contains at most one interior critical point.}
    
    Evaluating $H_{t,t}(\theta,\bm{\phi},\mathbf{u}_t)$ at $\theta^*(\bm{\phi},\mathbf{u}_t)$ yields
    \begin{equation}\label{eq:hdstar}
    	H_{t,t}(\theta^*,\bm{\phi},\mathbf{u}_t) =
    	\begin{cases} \frac{u_{pt}}{2}\left(1+ \sqrt{1+ R(\bm{\phi},\mathbf{u}_t)^2}\right) & \text{if $R(\bm{\phi},\mathbf{u}_t) < 0$} \\
    		\frac{u_{pt}}{2}\left(1 - \sqrt{1+ R(\bm{\phi},\mathbf{u}_t)^2}\right) & \text{if $R(\bm{\phi},\mathbf{u}_t) > 0$.}
    	\end{cases}
    \end{equation}
    To check whether $\theta^*(\bm{\phi},\mathbf{u}_t)$ delivers a minimum or a maximum of $H_{t,t}(\theta,\bm{\phi},\mathbf{u}_t)$, consider the second-order condition:
    \begin{equation}\label{eq:soc}
    	\frac{\partial^2 H_{t,t}(\theta,\bm{\phi},\mathbf{u}_t)}{\partial \theta^2} = - 2u_{pt}\cos(2\theta) - 2C(\bm{\phi},\mathbf{u}_t) \sin(2\theta).
    \end{equation}
    Substituting (\ref{eq:thetastar}) into (\ref{eq:soc}),
    \begin{equation}
    	\frac{\partial^2 H_{t,t}(\theta,\bm{\phi},\mathbf{u}_t)}{\partial \theta^2}\bigg|_{\theta = \theta^*(\bm{\phi},\mathbf{u}_t)} =
    	\begin{cases}
    		-2u_{pt} \sqrt{1+R(\bm{\phi},\mathbf{u}_t)^2} & \text{if $R(\bm{\phi},\mathbf{u}_t) < 0$} \\
    		2u_{pt} \sqrt{1+R(\bm{\phi},\mathbf{u}_t)^2} & \text{if $R(\bm{\phi},\mathbf{u}_t) > 0$.}
    	\end{cases}
    \end{equation}
    Since $C(\bm{\phi},\mathbf{u}_t) = R(\bm{\phi},\mathbf{u}_t)u_{pt}$, it follows that
    \begin{equation}\label{eq:soc_sign}
    	\mathrm{sign}\left(\frac{\partial^2 H_{t,t}(\theta,\bm{\phi},\mathbf{u}_t)}{\partial \theta^2}\bigg|_{\theta = \theta^*(\bm{\phi},\mathbf{u}_t)}\right) = \mathrm{sign}(C(\bm{\phi},\mathbf{u}_t))=
    	\begin{cases}
    		-1 & \text{if $C(\bm{\phi},\mathbf{u}_t) < 0$} \\
    		1 & \text{if $C(\bm{\phi},\mathbf{u}_t) > 0$.}
    	\end{cases}
    \end{equation}
    Consequently, if $\theta^*(\bm{\phi},\mathbf{u}_t) \in \operatorname{int} IS_{\theta}(\bm{\phi})$, it is a maximum when $C(\bm{\phi},\mathbf{u}_t) < 0$ and a minimum when $C(\bm{\phi},\mathbf{u}_t) > 0$.

    In evaluating $H_{t,t}(\theta,\bm{\phi},\mathbf{u}_t)$ at the end points of $IS_{\theta}(\bm{\phi})$, there are two cases to consider based on the sign of $\rho$. First, when $\rho < 0$,
    \begin{equation}\label{eq:hd_bound1_rhoneg}
    	H_{t,t}\left(\underline{\theta},\bm{\phi},\mathbf{u}_t\right) = \rho^2 \frac{\sigma_{11}}{\sigma_{21}}u_{qt} 
    \end{equation}
    and
    \begin{equation}\label{eq:hd_bound2_rhoneg}
    	H_{t,t}(\overline{\theta},\bm{\phi},\mathbf{u}_t) = u_{pt}. 
    \end{equation}
    When $\rho \geq 0$,
    \begin{equation}\label{eq:hd_bound1_rhopos}
    	H_{t,t}\left(\underline{\theta},\bm{\phi},\mathbf{u}_t\right) = 0 
    \end{equation}
    and
    \begin{equation}\label{eq:hd_bound2_rhopos}
    	H_{t,t}\left(\overline{\theta},\bm{\phi},\mathbf{u}_t \right) = (1-\rho^{2})\left(u_{pt} - \frac{\sigma_{21}}{\sigma_{22}}C(\bm{\phi},\mathbf{u}_t)\right). 
    \end{equation}
    Thus, one end point always delivers an extreme contribution to the observed forecast error. When $\rho < 0$, the supply shock may completely explain the price innovation. When $\rho \geq 0$, it may explain none of it.
    
    Because there is at most one interior critical point, the minimum and maximum must be attained among the values at $\underline{\theta}(\bm{\phi})$, $\overline{\theta}(\bm{\phi})$ and the admissible critical point $\theta^{*}(\bm{\phi},\mathbf{u}_t)$. This establishes (\ref{eq:thetaset})--(\ref{eq:is_htt_bounds}). 
    
\end{proof}

The following results present closed-form expressions for the bounds in Proposition~\ref{prop:htt_analytical}. Define $H_{t,t}^{*}(\bm{\phi},\mathbf{u}_t) \equiv H_{t,t}(\theta^{*}(\bm{\phi},\mathbf{u}_t),\bm{\phi},\mathbf{u}_t)$. There are four cases to consider depending on the sign of $\rho$ and whether the critical point $\theta^{*}(\bm{\phi},\mathbf{u}_t)$ lies within $IS_{\theta}(\bm{\phi})$.

\medskip

\noindent\textbf{Negative correlation, critical point admissible.} When $\rho < 0$, the critical point is admissible if and only if $\theta^{*}(\bm{\phi},\mathbf{u}_t) \geq \arctan(\sigma_{22}/\sigma_{21})$.\footnote{The second-order condition (\ref{eq:soc_sign}) indicates whether $\theta^{*}(\bm{\phi},\mathbf{u}_t)$ corresponds to the lower bound of $H_{t,t}(\theta,\bm{\phi},\mathbf{u}_t)$ or the upper bound, while the other bound can be obtained by comparing (\ref{eq:hd_bound1_rhoneg}) and (\ref{eq:hd_bound2_rhoneg}).} In this case, 
	\begin{equation}
		IS_{H_{t,t}}(\bm{\phi},\mathbf{u}_t) = 
		\begin{cases}
			\left[H_{t,t}^*(\bm{\phi},\mathbf{u}_t), \max\left\{u_{pt},\rho^2 \frac{\sigma_{11}}{\sigma_{21}}u_{qt}\right\}\right] & \text{if $C(\bm{\phi},\mathbf{u}_t) > 0$} \\
			\left[\min\left\{u_{pt},\rho^2 \frac{\sigma_{11}}{\sigma_{21}}u_{qt}\right\},H_{t,t}^*(\bm{\phi},\mathbf{u}_t)\right] & \text{if $C(\bm{\phi},\mathbf{u}_t) < 0$.} \\
		\end{cases}
	\end{equation}

    \medskip
    
  \noindent\textbf{Negative correlation, critical point inadmissible.}
	\begin{equation}
		IS_{H_{t,t}}(\bm{\phi},\mathbf{u}_t) = 
			\left[\min\left\{u_{pt},\rho^2 \frac{\sigma_{11}}{\sigma_{21}}u_{qt}\right\}, \max\left\{u_{pt},\rho^2 \frac{\sigma_{11}}{\sigma_{21}}u_{qt}\right\}\right].
	\end{equation}

    \medskip

\noindent\textbf{Nonnegative correlation, critical point admissible.} When $\rho \geq 0$, the critical point is admissible if and only if $\theta^{*}(\bm{\phi},\mathbf{u}_t) \leq \arctan(-\sigma_{21}/\sigma_{22})$.\footnote{The second-order condition (\ref{eq:soc_sign}) indicates whether $\theta^{*}(\bm{\phi},\mathbf{u}_t)$ corresponds to the lower bound of $H_{t,t}(\theta,\bm{\phi},\mathbf{u}_t)$ or the upper bound, while the other bound can be obtained by comparing (\ref{eq:hd_bound1_rhopos}) and (\ref{eq:hd_bound2_rhopos}).} In this case,
	\begin{equation}
		IS_{H_{t,t}}(\bm{\phi},\mathbf{u}_t) = 
		\begin{cases}
			\left[H_{t,t}^*(\bm{\phi},\mathbf{u}_t),\max\left\{0, (1-\rho^{2})\left(u_{pt} - \frac{\sigma_{21}}{\sigma_{22}}C(\bm{\phi},\mathbf{u}_t)\right)\right\}\right] & \text{if $C(\bm{\phi},\mathbf{u}_t) > 0$} \\		
			\left[\min\left\{0, (1-\rho^{2})\left(u_{pt} - \frac{\sigma_{21}}{\sigma_{22}}C(\bm{\phi},\mathbf{u}_t)\right)\right\}, H_{t,t}^*(\bm{\phi},\mathbf{u}_t)\right] & \text{if $C(\bm{\phi},\mathbf{u}_t) < 0$} \\
		\end{cases}
	\end{equation}	

 \medskip

 \noindent\textbf{Nonnegative correlation, critical point inadmissible.}
 	\begin{multline}
	IS_{H_{t,t}}(\bm{\phi},\mathbf{u}_t) = 
		\Bigg[\min\left\{0, (1-\rho^{2})\left(u_{pt} - \frac{\sigma_{21}}{\sigma_{22}}C(\bm{\phi},\mathbf{u}_t)\right)\right\}, \\
        \max\left\{0, (1-\rho^{2})\left(u_{pt} - \frac{\sigma_{21}}{\sigma_{22}}C(\bm{\phi},\mathbf{u}_t)\right)\right\}\Bigg].
	\end{multline}	

The following corollary uses the analytical characterisation above to describe properties of the identified set for the historical decomposition.

\begin{corollary}
    The identified set for $H_{t,t}$, $IS_{H_{t,t}}(\bm{\phi},\mathbf{u}_t)$, has the following properties:
    \begin{enumerate}
        \item If $\rho < 0$, then
        \begin{equation}
            u_{pt} \in IS_{H_{t,t}}(\bm{\phi},\mathbf{u}_t).
        \end{equation}
        If $\rho \geq 0$, then
        \begin{equation}
            0 \in IS_{H_{t,t}}(\bm{\phi},\mathbf{u}_t).
        \end{equation}
        \item If $\rho = 0$, then
        \begin{equation}\label{eq:rho0}
            [\min\{0,u_{pt}\},\max\{0,u_{pt}\}] \subseteq IS_{H_{t,t}}(\bm{\phi},\mathbf{u}_t).
        \end{equation}
        \item Along a sequence of reduced-form parameters such that $\rho \rightarrow -1$ and $\sigma_{11}/\sigma_{21} \rightarrow \kappa < \infty$,
        \begin{equation}
            IS_{H_{t,t}}(\bm{\phi},\mathbf{u}_t) \rightarrow   
                \left[\min\left\{u_{pt}, \kappa u_{qt}\right\}, \max\left\{u_{pt}, \kappa u_{qt}\right\}\right],
        \end{equation}
        in the sense that the lower and upper end points converge to the end points of this interval. Its limiting width is
        \begin{equation}\label{eq:hist_deviation}
            \left|u_{pt}- \kappa u_{qt}\right|.
        \end{equation}
    \end{enumerate}
\end{corollary}

\begin{proof}
    The first property follows from the end points in (\ref{eq:hd_bound2_rhoneg}) and (\ref{eq:hd_bound1_rhopos}). The second property follows from the end points in (\ref{eq:hd_bound1_rhopos}) and (\ref{eq:hd_bound2_rhopos}) evaluated at $\rho = 0$. For the third property, as $\rho \rightarrow -1$, $\sigma_{22}/\sigma_{21} \rightarrow 0$, so $\underline{\theta}(\bm{\phi})$ and $\overline{\theta}(\bm{\phi})$ both converge to zero. Under the assumption that $\sigma_{11}/\sigma_{21} \rightarrow \kappa$, this implies that the corresponding end points of  $IS_{H_{t,t}}(\bm{\phi},\mathbf{u}_t)$ converge to $u_{pt}$ and $\kappa u_{qt}$. 
\end{proof}

\medskip
\noindent\textbf{Remarks.} The first property can be interpreted as implying that the identified set always admits an `extreme' contribution. When $\rho < 0$, $u_{pt} \in IS_{H_{t,t}}(\bm{\phi},\mathbf{u}_t)$ at all values of the reduced-form parameters and at any realisation of the data; that is, the identified set always admits the possibility that $u_{pt}$ is entirely driven by the supply shock. Similarly, when $\rho \geq 0$, $0 \in IS_{H_{t,t}}(\bm{\phi},\mathbf{u}_t)$ and the identified set always admits the possibility that $u_{pt}$ is entirely driven by the demand shock, with the supply shock making no contribution.

The second property can be interpreted as implying that the sign restrictions are uninformative when the innovations are uncorrelated. By continuity, when $|\rho| \approx 0$, identified sets will also tend to admit contributions close to both zero and $u_{pt}$. When the price-quantity innovations are weakly correlated, the sign restrictions will therefore remain unable to rule out nearly pure demand- or supply-driven explanations of the price forecast error.\footnote{The identified set may extend beyond the interval between zero and $u_{pt}$. This can be shown directly using the analytical characterisation for $IS_{H_{t,t}}(\bm{\phi},\mathbf{u}_t)$. When $\rho = 0$, $C(\bm{\phi},\mathbf{u}_t) = (\sigma_{11}/\sigma_{22})u_{qt}$ and $R(\bm{\phi},\mathbf{u}_t) = (\sigma_{11}/\sigma_{22})(u_{qt}/u_{pt})$. Consider a case where $u_{pt}>0$ and $u_{qt} > 0$, so $IS_{H_{t,t}}(\bm{\phi},\mathbf{u}_t) = [H_{t,t}^*(\bm{\phi},\mathbf{u}_t),u_{pt}]$, where $H_{t,t}^*(\bm{\phi},\mathbf{u}_t) = \frac{u_{pt}}{2}\left(1 - \sqrt{1+ R(\bm{\phi},\mathbf{u}_t)^2}\right)$. Since $H_{t,t}^*(\bm{\phi},\mathbf{u}_t) < 0$, $[0,u_{pt}] \subset IS_{H_{t,t}}(\bm{\phi},\mathbf{u}_t)$.}

The third property implies that strong correlation between the innovations (i.e. $|\rho|\approx 1$) can make the sign restrictions highly informative, but only when $\mathbf{u}_t$ aligns with the historical price-quantity relationship. From (\ref{eq:hist_deviation}), the limiting width of the identified set is $|u_{pt} - \kappa u_{qt}|$, which measures the extent to which the realised forecast error departs from the historical (limiting) negative relationship between the innovations. As $\rho \rightarrow -1$, the population linear projection of $u_{pt}$ onto $u_{qt}$ converges to $\kappa u_{qt}$, which justifies interpreting (\ref{eq:hist_deviation}) as the deviation of $u_{pt}$ from its historical relationship with $u_{qt}$. If $u_{pt} = \kappa u_{qt}$, then $IS_{H_{t,t}}(\bm{\phi},\mathbf{u}_t)$ collapses to a point. Hence, the identified set will be highly informative. Conversely, strong correlation is not sufficient for the sign restrictions to deliver an informative identified set for the historical decomposition. For example, if $u_{pt} >0$ and $u_{qt} > 0$, the interval $[0,u_{pt}]$ will lie within $IS_{H_{t,t}}(\bm{\phi},\mathbf{u}_t)$. In that case, the identified set will admit the possibility that $u_{pt}$ was entirely driven by a supply shock or entirely by a demand shock, even though $\rho \approx -1$.

A similar argument applies as $\rho \rightarrow 1$. In that case, both end points of $IS_{\theta}(\bm{\phi})$ converge to $-\pi/2$. Consequently, the identified set can become narrow around zero, although its width again depends on whether the realised price and quantity innovations conform to their historical positive relationship. 

These results establish that $\rho$ determines the \emph{potential} identifying power of the sign restrictions, while the realised forecast error determines whether that potential is realised for a particular historical decomposition.

\renewcommand{\theequation}{C\arabic{equation}}
\renewcommand{\thesection}{C}
\renewcommand{\thefigure}{C\arabic{figure}}
\setcounter{equation}{0}
\setcounter{figure}{0}
\section{Beyond the Bivariate Model: An Example}
\label{app:beyondbivariate}

This appendix examines the identifying power of the supply-demand sign restrictions when adding an additional variable to the SVAR. Unlike the bivariate case, the identified set for $\mathbf{Q}$ cannot be indexed by a scalar set-identified parameter and the optimisation problem defining the bounds of the identified set for the historical decomposition does not admit a tractable analytical characterisation. Hence, I illustrate the identification problem numerically.

Consider a three-variable SVAR with $\mathbf{y}_t = (p_t,q_t,z_t)'$ and $\bm{\varepsilon}_t = (\varepsilon_{st},\varepsilon_{dt},\varepsilon_{zt})'$. Impose the following sign restrictions on $\mathbf{A}_0^{-1}$ and $\mathbf{A}_0$:
\begin{equation}
	\mathbf{A}_0^{-1} =
	\begin{bmatrix}
		+ & + & ? \\
		- & + & ? \\
		? & ? & ?
	\end{bmatrix}, \quad \mathbf{A}_0 =
	\begin{bmatrix}
		+ & - & ? \\
		+ & + & ? \\
		? & ? & ?
	\end{bmatrix}
	,
\end{equation}
where $?$ indicates that the element is unrestricted. These restrictions are consistent with the restrictions imposed in the bivariate model. Imposing restrictions on both $\mathbf{A}_0^{-1}$ and $\mathbf{A}_0$ is necessary here, because -- unlike in the bivariate model -- imposing the sign restrictions on $\mathbf{A}_0^{-1}$ only does not guarantee that the first two structural equations represent supply and demand curves.

Let $\rho_{ij} = \mathrm{Corr}(u_{it},u_{jt})$ be the correlation between the one-step-ahead forecast errors in variables $i$ and $j$. I specify a data-generating process such that $\rho_{12} \approx -1$, $\rho_{13}$ is small and positive and $\rho_{23}$ is small and negative. Suppose we observe a sample of 100~one-step-ahead forecast errors from this process, illustrated in Figure~\ref{fig:trivariate_scatter}. The strongly negatively correlated forecast errors in $p_t$ and $q_t$ are consistent with the bivariate example in the main text. In the bivariate model, when the reduced-form innovations are strongly negatively correlated, the identified sets for $H_{t,t}$ will tend to be narrow and unambiguously imply that unexpected changes in $p_t$ were predominantly driven by supply shocks. 

\begin{figure}[h]
	\centering
	\caption{One-step-ahead Forecast Errors in Three-variable Model} \label{fig:trivariate_scatter}
	\includegraphics[scale=0.6]{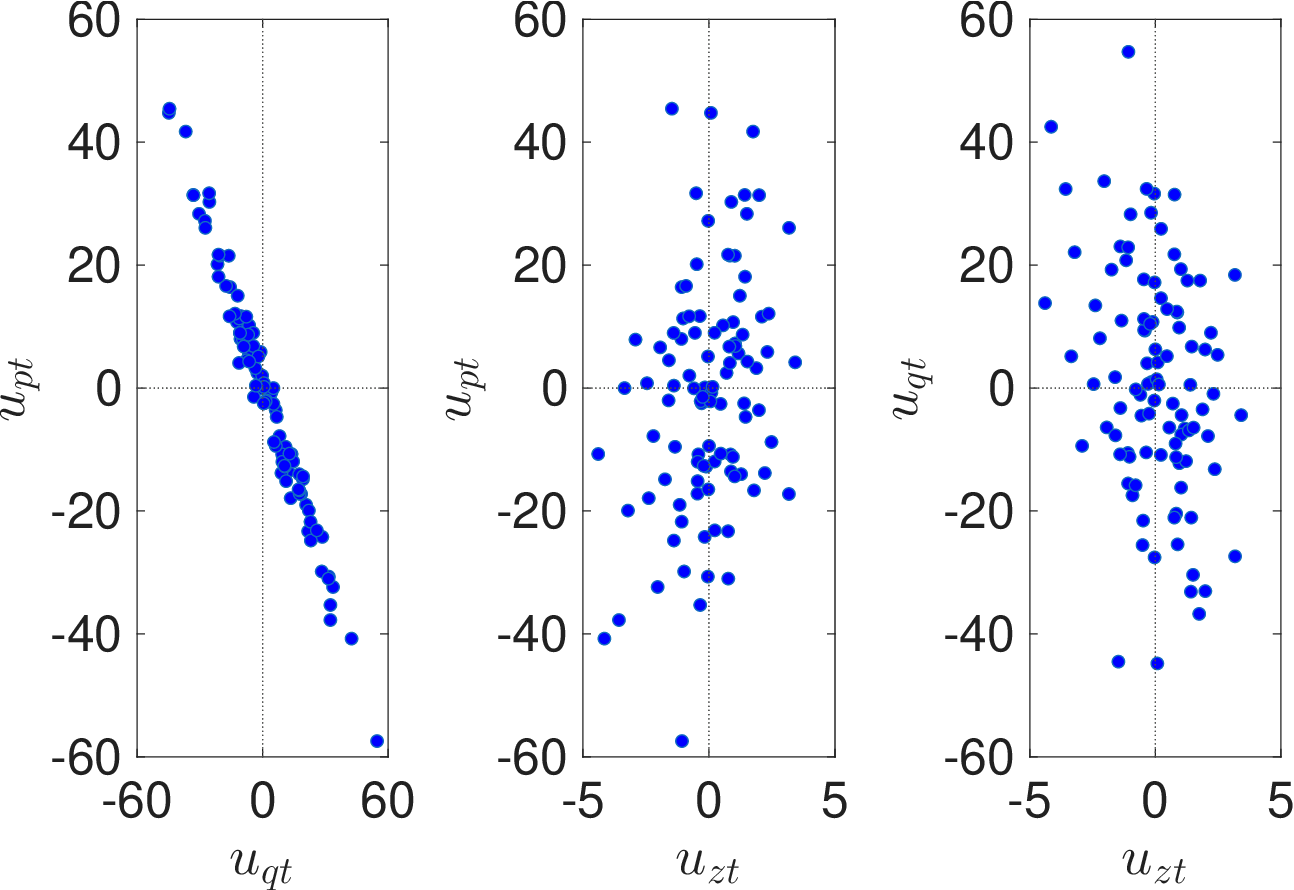}
    \noindent\begin{minipage}{\textwidth} 
	\small\textbf{Notes:} Realisations of $\mathbf{u}_t$ generated from three-variable VAR.
    \end{minipage}
\end{figure}

Figure~\ref{fig:trivariate_is} plots the 100 realisations of $u_{pt}$ along with the identified sets for $H_{t,t}$.\footnote{The identified sets are approximated by computing the minimum and maximum of $H_{t,t}$ over uniform draws of $\mathbf{Q}$ from within its identified set. The draws of $\mathbf{Q}$ are obtained using the algorithm from \cite{Rubio-Ramirez_Waggoner_Zha_2010}. Based on the results in \cite{Montiel-Olea_Nesbit_2021}, I use around $55,000$ draws to approximate the identified sets, which guarantees misclassification error of at most 1~per cent with probability at least 99~per cent.} In stark contrast with the bivariate case, these identified sets always include zero and the realisation of $u_{pt}$, implying that the sign restrictions are uninformative about the contribution of the supply shock. Intuitively, this lack of informativeness occurs because we are introducing an additional (unidentified) shock, which -- in the language of \cite{Wolf_2020} -- is able to `masquerade' as the supply shock. Adding variables to the system without imposing additional restrictions is therefore unlikely to aid in identifying historical decompositions.

\begin{figure}[h]
	\centering
	\caption{Identified Sets for Historical Decompositions in Three-variable Model} \label{fig:trivariate_is}
	\includegraphics[scale=0.6]{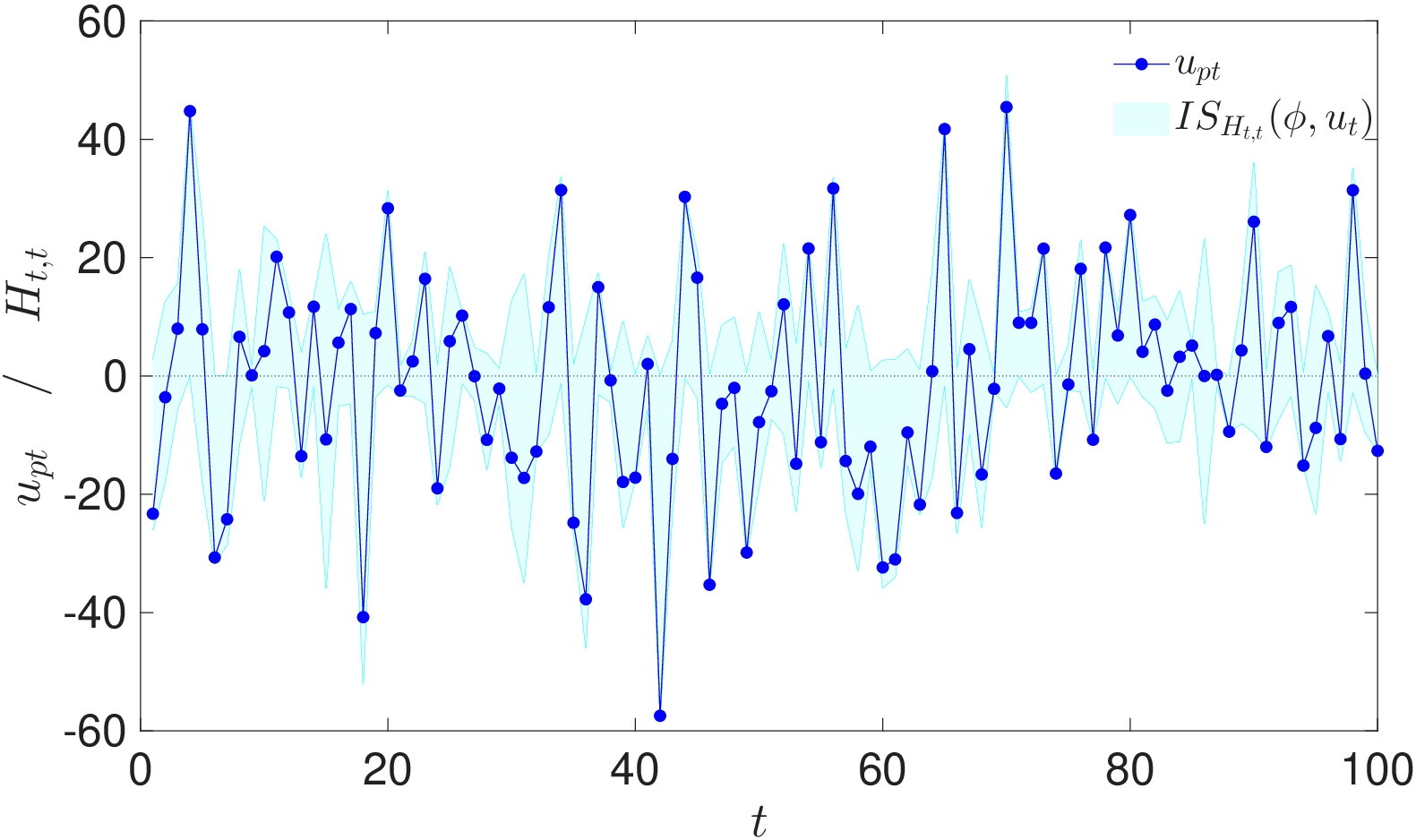}
    \noindent\begin{minipage}{\textwidth} 
	\small\textbf{Notes:} Innovations in $p_t$ simulated from data-generating process with three variables. Shaded region represents identified set for contribution of supply shock at each $t$, $H_{t,t}$, given sign restrictions.
    \end{minipage}
\end{figure}

\setcounter{section}{0}
\renewcommand{\thesection}{D}
\section{Additional Empirical Details}
\label{app:additionalresults}
\renewcommand{\theequation}{D.\arabic{equation}}
\renewcommand{\thefigure}{D\arabic{figure}}
\setcounter{equation}{0}
\setcounter{figure}{0}

This appendix describes estimation of historical decompositions in the empirical applications. Appendix~\ref{app:aggregate_imp} describes estimation under the uniform prior used by \cite{Giannone_Primiceri_2024_nber} (henceforth, GP24) and the prior-robust Bayesian procedure used to estimate identified sets and quantify reduced-form parameter uncertainty. Appendix~\ref{app:euro} presents results from decomposing euro area inflation. Appendix~\ref{app:larger_svars} presents results based on the three-variable SVARs considered in GP24. Appendix~\ref{app:disaggregated_det} provides additional details about the decompositions of disaggregated PCE data.

\subsection{Details about aggregate decompositions}
\label{app:aggregate_imp}

I estimate exactly the same reduced-form VAR as in GP24. The model includes four lags of the endogenous variables. Estimation is carried out via Bayesian methods, with a Minnesota and sum-of-coefficients prior, following \cite{Giannone_Lenza_Primiceri_2015}. As emphasised in GP24, the sum-of-coefficients prior reduces posterior uncertainty about the model's deterministic component, which can be an important source of uncertainty about historical decompositions (\citealt{Bergholt_etal_2024}). GP24 additionally assume a uniform prior for $\mathbf{Q}$ (equivalently, $\theta$) and report posterior means for the historical decompositions. When replicating their results, I obtain 10,000 draws of the reduced-form parameters $\bm{\phi}$ and 10~draws of $\theta$ from the uniform distribution over $IS_{\theta}(\bm{\phi})$ at each draw of $\bm{\phi}$, yielding 100,000 draws of $(\bm{\phi}',\theta)$. 

When estimating identified sets, I quantify posterior uncertainty around the estimates using the prior-robust Bayesian approach to inference from \cite{Giacomini_Kitagawa_2021}. This approach involves replacing the conditional prior for $\theta$ with the class of all conditional priors that are consistent with the identifying restrictions (i.e. that assign probability one to $IS_{\theta}(\bm{\phi})$). This generates a class of posteriors for the parameters of interest, which can be summarised using different outputs. The `set of posterior means' for a scalar object of interest (e.g. the historical decomposition in a given period) is an interval spanning the posterior means corresponding to the class of posteriors; this can be interpreted as an estimator of the identified set. A `robust credible interval' with credibility $1-\tau$ is an interval that is assigned at least posterior probability $1-\tau$ under any posterior in the class; these intervals naturally represent posterior uncertainty about the identified sets. 

In practice, approximating the set of posterior means and robust credible intervals requires computing identified sets at every draw of $\bm{\phi}$ from its posterior. In the current setting, this is computationally simple given the numerical procedure for computing the bounds of the identified sets for the decompositions described in Online Appendix~\ref{sec:computation}. For a given scalar parameter of interest $\eta$, the end points of the set of posterior means are approximated by the posterior means of the end points of the identified set for $\eta$. The $1-\tau$ robust credible interval is approximated by an interval with lower (upper) end point equal to the $\tau/2$ ($1-\tau/2$) posterior quantile of the lower (upper) end point of the identified set for $\eta$.

\subsubsection{Joint inference}
\label{app:jointinference}

The analysis in the main text focuses on scalar historical decompositions as the structural objects of interest, with inference conducted marginally for each shock contribution and period. In some applications, researchers may be interested in hypotheses that involve multiple shock contributions (e.g. \citealt{Inoue_Kilian_2022}); for example, relating to the contributions of different shocks in the same period or the contributions of the same shock in different periods. Examining such hypotheses requires conducting joint inference on vector-valued objects. However, the associated joint identified sets may be difficult to visualise and summarise. As discussed in \citet{Giacomini_Kitagawa_Read_2022_rejoinder}, the robust Bayes approach can be used to assess the evidence in support of hypotheses that involve multiple parameters using \emph{posterior lower and upper probabilities}; these are, respectively, the smallest and largest posterior probabilities assigned to a specified hypothesis attainable within the class of posteriors. 

I make use of these ideas to quantify the evidence for the hypothesis that demand shocks were the dominant driver of the US inflation surge. More specifically, I consider the hypothesis that, in every quarter from 2021:Q2 to 2022:Q2, the contribution of the demand shock to the year-ended CPI inflation forecast error (given information up to 2019:Q4) exceeded the contribution of the supply shock. Since the contributions of supply and demand shocks sum to the forecast error, this hypothesis can be formulated as $\mathcal{H} = \bigcup_{t \in \mathcal{T}} \{\tilde{H}_{t} \leq (1/2)FE_t\}$, where $\tilde{H}_{t}$ is the contribution of the supply shock to the year-ended inflation forecast error in period $t$ given information up to 2019:Q4 and $FE_t$ is the forecast error.\footnote{To approximate the posterior (lower and upper) probabilities, I obtain 10,000 draws of $\theta$ from the uniform distribution over $IS_{\theta}(\bm{\phi})$ at every draw of $\bm{\phi}$; see \cite{Giacomini_Kitagawa_Read_2022_rejoinder} for details about how the posterior lower and upper probabilities are approximated.}

Under the uniform prior, the posterior probability assigned to $\mathcal{H}$ is around 0.5, indicating there is a roughly even chance that the inflationary impact of demand shocks outweighed that of supply shocks in every period during the inflation surge. However, the posterior lower probability is zero and the posterior upper probability is 0.99. The sign restrictions are therefore uninformative about the hypothesis. This exercise indicates that the sign restrictions remain uninformative about the drivers of the inflation surge when considering joint, rather than marginal, inference.

\subsection{Euro area results}
\label{app:euro}

GP24 decompose euro area consumer price inflation using the same model that is used to decompose US inflation (described in Section~\ref{subsec:agg}). Consumer price inflation is based on the Harmonized Index of Consumer Prices (HICP). Figure~\ref{fig:GP24_EA} plots the estimated supply contributions under both the uniform prior and the robust Bayesian approach. The posterior mean under the standard Bayesian approach to inference implies that, at its peak in 2022:Q4, the forecast error in year-ended inflation was largely attributable to demand shocks; supply shocks contributed only 2.4~percentage points of the 7.9~percentage point forecast error. However, the set of posterior means for this supply contribution ranges from around zero to 6.8~percentage points, so the sign restrictions are again largely uninformative about the magnitude of the supply and demand contributions.

\begin{figure}[!htp]
	\centering
	\caption{Contribution of Supply Shocks to Year-ended Euro Area HICP Inflation}\label{fig:GP24_EA}
         \includegraphics[scale=0.6]{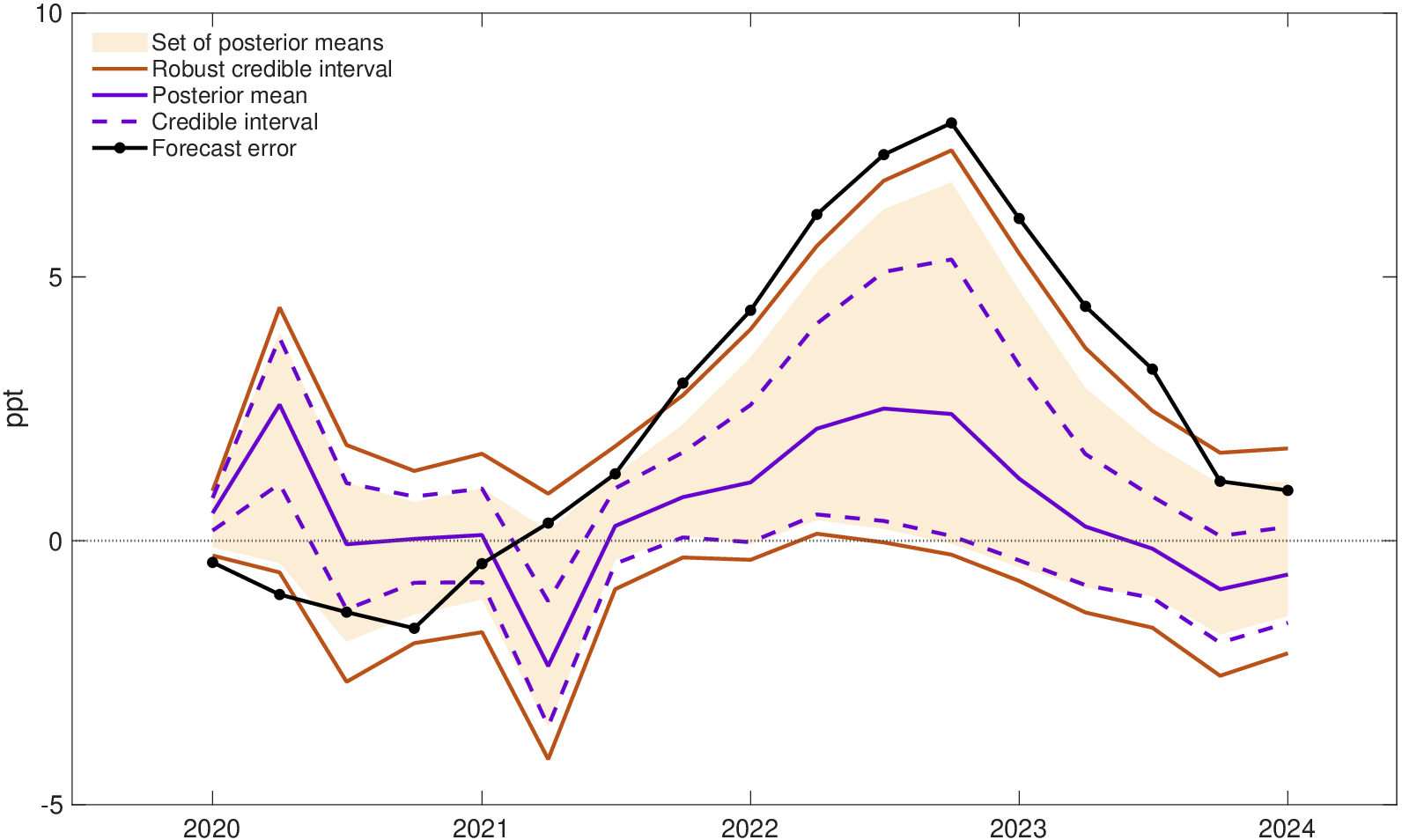}
    \noindent\begin{minipage}{\textwidth}
	\small\textbf{Notes:} Contribution of supply shocks to forecast error in year-ended HICP inflation given information available up to 2019:Q4. Set of posterior means is interpretable as estimator of identified set. Standard and robust credible intervals are at the 68~per cent credibility level.
     \end{minipage}  
\end{figure}

\subsection{Larger SVARs}
\label{app:larger_svars}

In addition to decompositions based on bivariate SVARs, GP24 estimate two distinct decompositions based on three-variable SVARs. This section estimates identified sets for these decompositions. Since these are three-variable SVARs, the numerical approach to computing the end points of the identified set used above is no longer applicable. I therefore approximate the end points of identified sets using 10,000~draws of $\mathbf{Q}$ from a uniform distribution over its identified set obtained via accept-reject sampling (\citealt{Rubio-Ramirez_Waggoner_Zha_2010}).\footnote{The results in \cite{Montiel-Olea_Nesbit_2021} suggest that using around 10,000 draws of $\mathbf{Q}$ to approximate the end points of the identified sets is sufficient to ensure an accurate approximation. According to their Theorem~3, the number of draws required from inside the identified set to guarantee misclassification error less than $\epsilon$ occurs with probability at least $1-\delta$ is $\mathrm{min}\{2d\ln(2d/\delta),\exp(1)(2d+\ln(1/\delta))\}/\epsilon$, where $d$ is the dimension of the parameter region being approximated (i.e. the number of historical decompositions). In the current exercise, setting $d = 17$ (the number of post-2019 quarters) and $\epsilon = \delta = 0.01$ yields $K=10,494$.}

\begin{figure}[htp]
	\centering
	\caption{Contributions of Shocks to Year-ended Non-energy CPI Inflation}\label{fig:GP24_Energy_US}
         \includegraphics[scale=0.7]{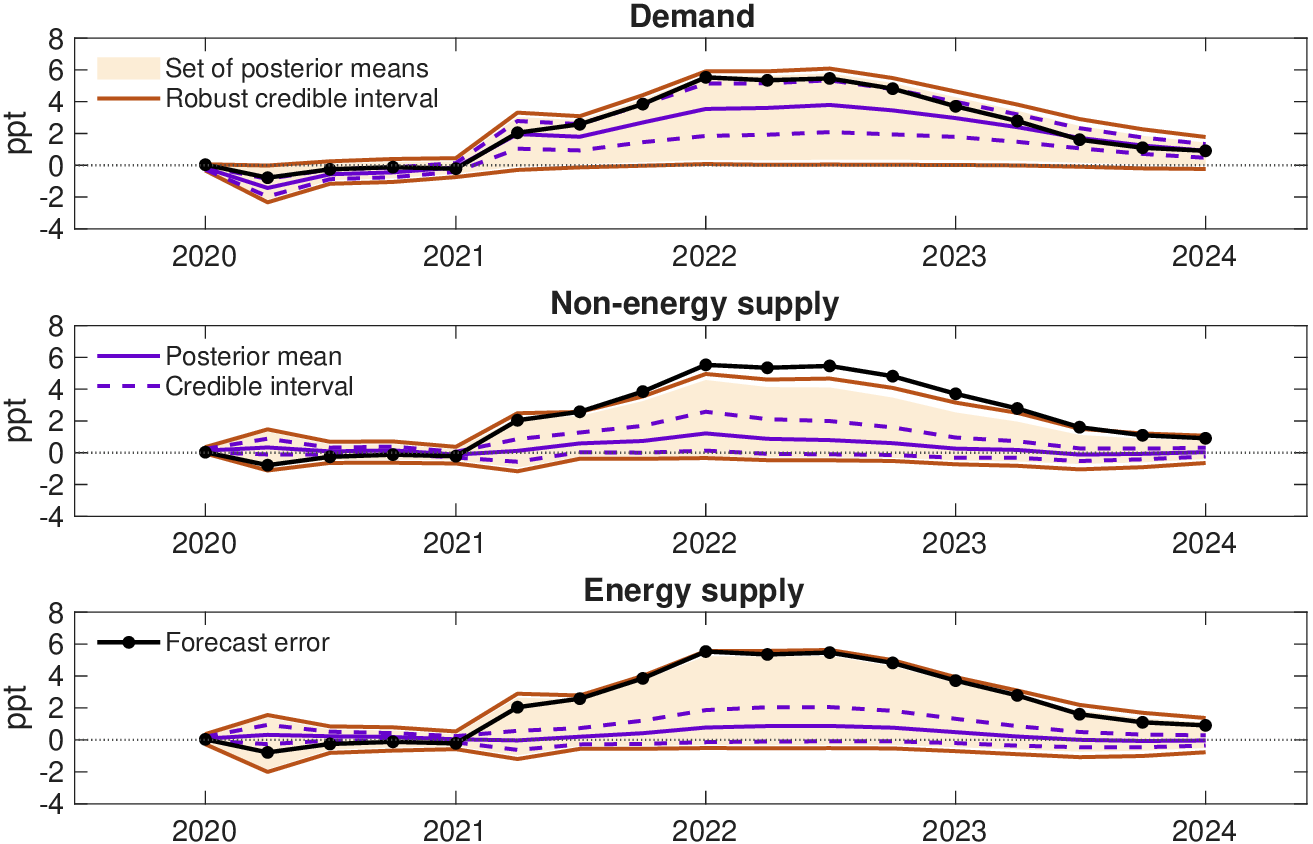}
    \noindent\begin{minipage}{\textwidth}
	\small\textbf{Notes:} Contributions of shocks to forecast error in year-ended non-energy CPI inflation given information available up to 2019:Q4. Set of posterior means is interpretable as estimator of identified set. Standard and robust credible intervals are at the 68~per cent credibility level.
     \end{minipage}  
\end{figure}

The first three-variable SVAR replaces inflation in the baseline model with two series measuring energy and non-energy prices, and identifies three separate shocks: (i) demand shocks, which are assumed to move GDP, energy and non-energy prices in the same direction; (ii) non-energy supply shocks, which are assumed to move non-energy prices in the opposite direction to GDP and energy prices; and (iii) energy supply shocks, which are assumed to move energy prices in the opposite direction to GDP (the response of non-energy prices is left unrestricted). 

Figure~\ref{fig:GP24_Energy_US} plots the contributions of each of the three shocks to the post-2019 forecast error in year-ended US non-energy inflation. The results based on the standard Bayesian approach to inference with a uniform prior are consistent with those in GP24; in particular, the contribution of energy supply shocks is substantially smaller than that of demand shocks. However, in all periods, the set of posterior means for the contribution of the demand shock spans values ranging from around zero to the realisation of the forecast error. The sign restrictions are therefore largely uninformative about the extent to which (non-energy) inflation in the US was driven by demand or supply shocks in this period.

The second three-variable SVAR that GP24 consider augments their baseline model with a measure of interest rates to capture the monetary policy stance. They identify three shocks: (i) demand shocks, which are assumed to move GDP, prices and the nominal interest rate in the same direction; (ii) supply shocks, which are assumed to move GDP and prices in opposite directions (with the interest rate response unrestricted); and (iii) monetary policy shocks, which are assumed to move the nominal interest rate in the opposite direction to GDP and prices. The sign of the interest-rate response to a monetary policy shock is constrained for four periods rather than only on impact.

\begin{figure}[htp]
	\centering
	\caption{Contributions of Shocks to Year-ended US CPI Inflation -- Monetary Model}\label{fig:GP24_MP_US}
         \includegraphics[scale=0.7]{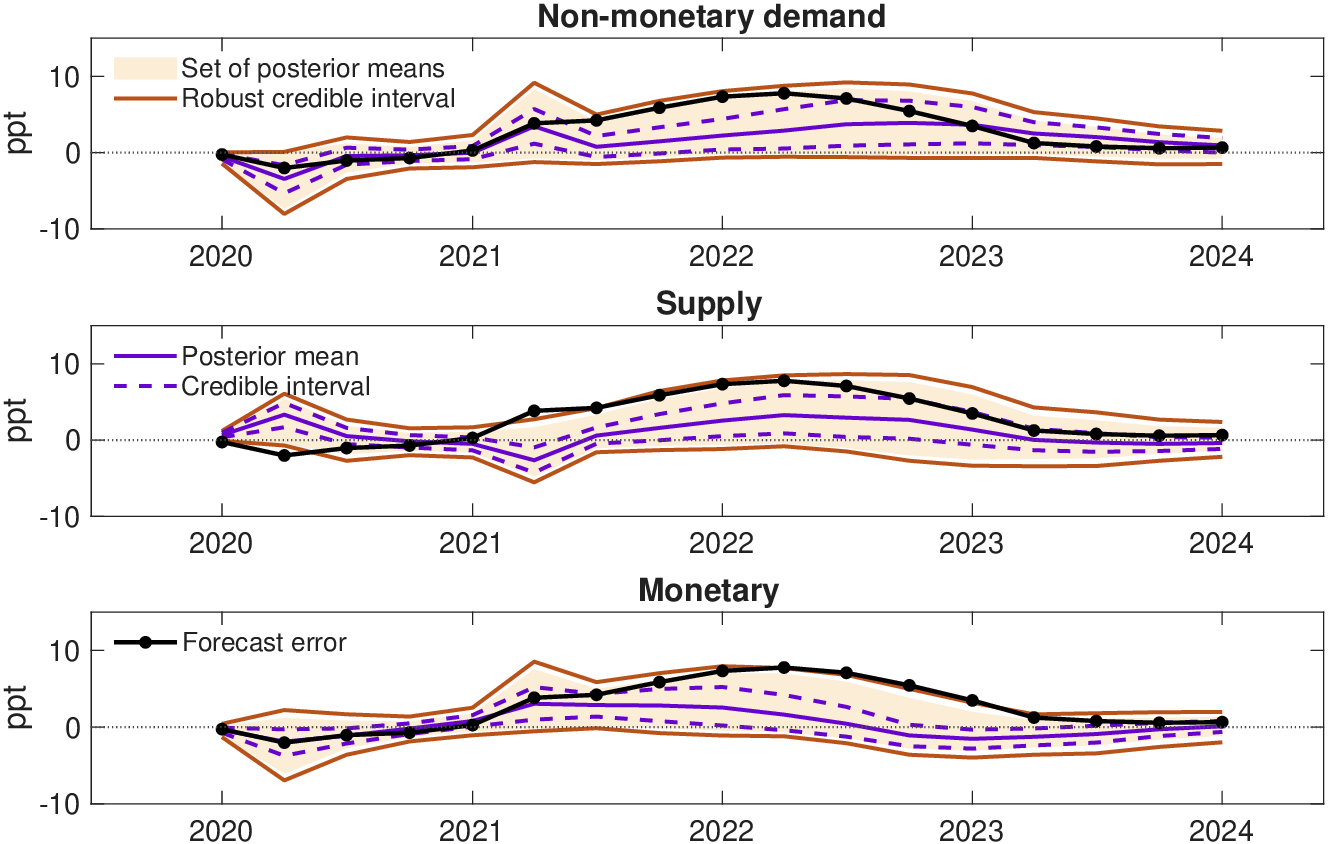}
    \noindent\begin{minipage}{\textwidth}
	\small\textbf{Notes:} Contributions of shocks to forecast error in year-ended CPI inflation given information available up to 2019:Q4. Set of posterior means is interpretable as estimator of identified set. Standard and robust credible intervals are at the 68~per cent credibility level.
     \end{minipage}  
\end{figure}

Figure~\ref{fig:GP24_MP_US} plots the contributions of each of the three shocks to the post-2019 forecast error in year-ended US CPI inflation. Again, the set of posterior means and robust credible intervals indicate that the sign restrictions are largely unable to distinguish between demand- and supply-driven explanations for the large run-up and subsequent decline in US inflation.

\subsection{Details about disaggregated decompositions}
\label{app:disaggregated_det}

The specifications of the VARs for each expenditure category follow the specifications in \cite{Shapiro_2026}, with two main exceptions.\footnote{The data are from the Bureau of Economic Analysis National Accounts Underlying Detail Tables.} First, the baseline specification in \cite{Shapiro_2026} is in levels, whereas mine is in log first differences. I work with the first-difference specification because the VARs in log levels are often explosive, in which case the historical decompositions of interest are also explosive.\footnote{In the first-difference specification, the VAR is explosive in two categories.} Second, \cite{Shapiro_2026} estimates the VARs using rolling OLS regressions, whereas I estimate constant-parameter VARs for simplicity. 

The VAR for each category is specified in terms of monthly PCE inflation (approximated by log differences). In the terminology of Section~\ref{subsec:historicaldecomp}, the concept of the historical decomposition that I consider is $H_{1,t}$. This represents the contribution of \emph{all} supply shocks occurring in the sample period to the realisation of monthly inflation in period $t$. I present the contribution of supply shocks to the realisation of \emph{year-ended} inflation. To obtain identified sets for these historical decompositions, I follow the strategy described in Online Appendix~\ref{sec:computation}. In particular, given that monthly inflation is approximated by log differences, year-ended inflation is equal to the rolling twelve-monthly sum of monthly inflation: $\pi_t^{(k,ye)} = \sum_{l=0}^{11}\pi_{t-l}^{(k)}$. Given that $H_{1,t}$ is the contribution of supply shocks to monthly inflation in period $t$, the contribution of supply shocks to year-ended inflation is then:
\begin{equation}\label{eq:hd_ye_dis}
    \sum_{h=0}^{11}H_{1,t-h}(\theta,\bm{\phi},\left\{\mathbf{u}_{l}\right\}_{l=1}^{t-h}) =  \begin{bmatrix}
        \cos\theta & \sin\theta
    \end{bmatrix}
    \tilde{\bm{\Omega}}_{t-1}(\bm{\phi}, \left\{\mathbf{u}_{l}\right\}_{l=1}^{t})
    \begin{bmatrix}
        \cos\theta \\
        \sin\theta
    \end{bmatrix},
\end{equation}
where
\begin{equation}
    \tilde{\bm{\Omega}}_{t-1}(\bm{\phi}, \left\{\mathbf{u}_{l}\right\}_{l=1}^{t}) = \sum_{h=0}^{11}\bm{\Omega}_{t-h}(\bm{\phi}, \left\{\mathbf{u}_{l}\right\}_{l=1}^{t-h}).
\end{equation}
The objective in (\ref{eq:hd_ye_dis}) has the same structure as (\ref{eq:hdtheta}). Hence, the numerical procedure from Online Appendix~\ref{sec:computation} can be used to compute the end points of the identified set, replacing $\bm{\Omega}_{h}(\bm{\phi}, \left\{\mathbf{u}_{l}\right\}_{l=t}^{t+h})$ with $\tilde{\bm{\Omega}}_{t-1}(\bm{\phi}, \left\{\mathbf{u}_{l}\right\}_{l=1}^{t})$.

\putbib
\end{bibunit}

\end{document}